\documentclass[10pt,fleqn]{article}
\usepackage{psfrag}
\usepackage{amsmath}
\usepackage[dvips]{epsfig}
\usepackage{epsfig}
\usepackage{setspace}
\usepackage{amsmath,amssymb}
\usepackage{amsthm}
\usepackage{mathtools}
\usepackage{graphics}
\usepackage[pdftex,colorlinks,bookmarks,pdfpagemode=UseOutlines,linkcolor=black,
pagecolor=black,urlcolor=black,citecolor=black,letterpaper,pageanchor=false]{hyperref}
\usepackage[margin=1in]{geometry}
\usepackage[T1]{fontenc}
\usepackage{titlesec}

\newtheorem{theo}{Theorem}[section]
\newtheorem{lemma}{Lemma}[section]
\newtheorem{df}{Definition}[section]
\newtheorem{cor}{Corollary}[section]
\newtheorem{assump}{Assumption}[section]
\newtheorem{assert}{Assertion}[section]
\newtheorem{remark}{Remark}[section]
\newcommand{\bl}{\begin{lemma}}
\newcommand{\el}{\end{lemma}}
\newcommand{\be}{\begin{equation}}
\newcommand{\ee}{\end{equation}}
\newcommand{\beqn}{\begin{eqnarray}}
\newcommand{\eeqn}{\end{eqnarray}}
\newcommand{\bt}{\begin{theo}}
\newcommand{\et}{\end{theo}}
\newcommand{\bd}{\begin{df}}
\newcommand{\ed}{\end{df}}
\newcommand{\ba}{\begin{assump}}
\newcommand{\ea}{\end{assump}}
\newcommand{\bass}{\begin{assert}}
\newcommand{\eass}{\end{assert}}
\newcommand{\brem}{\begin{remark}}
\newcommand{\erem}{\end{remark}}
\newcommand{\bc}{\begin{cor}}
\newcommand{\ec}{\end{cor}}

\numberwithin{equation}{section}

\long\def\comment#1{}

\DeclareMathOperator*{\argmax}{arg\,max}

\title{Modeling Toll Lanes and Dynamic Pricing Control
\author{Elena G. Dorogush and Alex A. Kurzhanskiy}
}

\begin{document}
\maketitle


\begin{abstract}
In this paper we address the problem of dynamic pricing for toll lanes
on freeways.
The proposed toll mechanism is broken up into two parts:
(1) the supply side feedback control that computes the desired 
split ratios for the 
incoming traffic flows between the general purpose and the toll lanes;
and (2) the demand side price setting algorithm that aims to enforce
the computed split ratios.

The split ratio controller is designed and tested in the context of
the link-node Cell Transmission Model with the modified node model
of in/out flow distribution.
The equilibrium structure of this traffic model is presented; and the
case, in which the existence of a toll lane is meaningful, is discussed.

For the price setting, two alternative approaches are presented.
The first one is commonly used, and it relies on the known 
Value of Time (VoT) distribution.
Its shortcoming, however, is in the difficulty of the VoT distribution
estimation.
The second approach employs the auction mechanism, where travelers make
bids on places in the toll lane.
The advantage of this approach is that it enables direct control over
how many vehicles will be allowed into the toll lane.
\end{abstract}


\section{Introduction}\label{sec_intro}
Congestion pricing is an economics concept of using pricing mechanisms to
charge the users of public roads for the negative effect on others generated
by the demand in excess of available supply.
It is one of a number of demand side traffic management strategies that
address traffic congestion.
Other known demand management strategies include: parking restrictions;
park and ride facilities allowing parking at a distance and continuation
by public transport or ride sharing;
reduction of road capacity to force traffic onto other travel modes;
road space rationing, where regulatory restrictions prevent certain types of 
vehicles from driving under certain circumstances or in certain areas; and
policy approaches, which encourage greater use of existing alternatives through
promotion, subsidies or restrictions.

Our focus will be on congestion pricing for Express Toll Lanes (ETL) 
and High Occupancy/ Toll (HOT) lanes.
This type of facilities becomes ever more popular in the U.S.
In San Francisco Bay Area alone, the Metropolitan Transportation Commission
promises the implementation of 550-mile express lane network by 2035,
and all of it with dynamic pricing strategies.

There were numerous studies on congestion pricing over the past two decades.
Optimal tolls under stochastic user equilibria were devised in \cite{smith94}.
Deterministic static equilibrium model for urban transport networks
with elastic demand and capacity constraints was presented in \cite{ferrari95}.
Optimization of tunnel tolls in Hong Kong was investigated in \cite{yanglam96}.
Yang and Bell talked about road pricing in the presence of congestion
delay \cite{yangbell97}.
Other models of congestion toll were discussed in \cite{hearnramana98}.
Wie and Tobin presented two pricing models --- one
static with the same day-to-day demand and road network capacity, and
the other one with time varying traffic conditions and congestion
tolls \cite{wietobin98}.
Bottleneck model with elastic demand was described in \cite{arnott98}.
The model of bottleneck congestion was generalized and optimal peak-load
toll was derived in \cite{mun99}.
Eliasson suggested optimal road pricing policy reducing aggregate 
travel time and distributing toll burden equally among travelers, although
travelers had utlity functions with constant marginal utilities of time
and money, and these marginal utilities were unobservable \cite{eliasson01}.
The second best pricing policy for a static transportation network, where
not all links could be tolled, was presented in \cite{verhoef02}.
De Palma and Lindsey derived the optimal tolls formulae and analyzed them to
reveal the separate influences of traveler heterogeneity, road network
effects, fiscal effects and equity concerns \cite{depalma04}.
In \cite{verhoefsmall04} the properties of various types of public and
private pricing on a congested road network were explored, and it was
shown that welfare-maximizing pricing was more efficient than the
revenue-maximizing one.
Game theoretical analysis of congestion pricing from its micro-foundations,
interaction of two or more travelers, was performed in \cite{levinson05}.
Simulation based analysis of various road pricing schemes was conducted in
\cite{depalma05}.
The analytical formulation of the optimal toll calculation for multi-class
traffic was given in \cite{holguin-veras09}.
In \cite{dong11} two dynamic pricing strategies, reactive (feedback) and
proactive (feedforward) were considered and compared, favoring the latter.
A comprehensive review of the current congestion pricing technologies
was offered in \cite{depalma11}.

It is important to mention the research targeted specifically at the dynamic
pricing of managed toll lanes.
A local feedback strategy for setting tolls is proposed in \cite{zhang08}.
The control algorithm is similar to ALINEA \cite{alinea},
where the portion of vehicle flow allowed into the toll lane for a given
time step is computed from the toll vehicle flow at the previous time step
and the speed observed in the toll and the general purpose lanes.
In the situation when the speed in the toll lane drops below 45 mph,
the congestion in the toll lane is assumed, and tolls are increased quicker
than otherwise.
Proactive dynamic pricing methodology is described in \cite{michalaka09}.
Tolls are determined based on the cell transmission model (CTM) with stochastic
demand and capacity.
The road configuration with a single decision point is considered.
A distance-based dynamic pricing strategy for toll lanes aimed at maximizing
toll revenue is presented in \cite{yang12}.
Here the tolls are set based on the travel time prediction that is
computed using
the stochastic variation of the LWR model \cite{chu11, lighthill55, richards56}.
The optimization problem maximizing the expected toll revenue with constraints
keeping the toll lane in free flow is solved. 
Yin and Lou delivered the proof of concept of a reactive
self-learning approach for determining time-varying
tolls in response to the detected traffic arrivals \cite{yinlou09}.
The approach learns in a sequential fashion motorists' willingness to pay
and then determines pricing strategies based on a point-queue model.
This result was extended in \cite{lou11}, where the CTM 
was used instead of the point-queue model,
and tolls are set based on the predicted traffic state.
These papers assume that toll lane should be kept in free flow, and the
corresponding constraints are imposed.
The logit model is used to implement the driver's lane choice.

In this paper, we approach dynamic pricing control from both, the supply
and the demand sides with the main focus on the supply side.
The supply side control is considered in the context of the
Cell Transmission Model (CTM) \cite{daganzo95} with the node model
proposed in \cite{tampere11}.
Here we extend the result of \cite{gomes08}
by analyzing the equilibrium structure of the freeway traffic model
with a single mainline, similarly to \cite{gomes08}, but for a different
node model.
Then, we discuss when it makes sense to split the mainline into two
lanes, the general purpose and the toll lanes, where the toll
lane is supposed to be less congested than the general purpose one.
It is important to note that activating a toll lane is not always
beneficial for the freeway performance.
Obviously, when the traffic density is low, a toll lane does not provide any
advantage in terms of travel time and, thus, should be free.
On the other hand, when the freeway is too congested, maintaining a toll lane
in free flow may be harmful, as it will lead to large congestion spillback
and degrade the throughput.
We propose a supply side control algorithm that computes split ratios dividing
the incoming flows between the general purpose and the toll lanes.
Its goal is to keep the toll lane in free flow as long as possible, then,
in the case of continuing excessive demand, allow the toll lane to congest
up to the point when the toll lane becomes a general purpose lane
in terms of its performance.
When the congestion recedes, the control algorithm ensures that the toll
lane frees up first.
The objective is to maintain the same throughput as in the freeway with all
lanes being general purpose.
The idea of this control algorithm is inspired by the HERO coordinated
ramp metering \cite{papamichail10}.
While in the case of ramp metering ramps are used as storage for the excess
demand, the general purpose lane serves as the storage for the split ratio 
control.
All proofs are given in the context of the proposed traffic model.
Once the desired flow splits are computed by the split ratio controller,
we achieve them by setting the toll rate --- that is already a demand
management.
The common way of setting the price is to rely on the Value of Time (VoT)
distribution \cite{brownstone05}, which is assumed to be known.
Relying on the VoT distribution, though, makes our control indirect, as we can only
estimate how much travelers are ready to pay.
The alternative solution is to use the truthful auction, which we propose
as a means of ensuring the desired split ratios.

The rest of the paper is organized as follows.
Section \ref{sec_model} introduces the traffic model, describes its
equilibrium structure and splitting the mainline into
the general purpose and the toll lanes.
Section \ref{sec_control} presents the split ratio control algorithm and
provides examples.
Section \ref{sec_tolls} is dedicated to pricing mechanisms, discussing the
estimation of VoT distribution and proposing the auction algorithm as an alternative.
The simulation of Interstate 680 South in Contra Costa County in California
is discussed in Section \ref{sec_application}.
Ultimately, Section \ref{sec_conclusion} concludes the paper.

\section{Traffic Model}\label{sec_model}
A road network is modeled by a directed graph~$G = (E, V)$,
where $E$ is the set of edges (links), $V$ is the set of vertices (nodes).
A link $e \in E$ is an ordered pair of nodes:
$e = (u, v)$, $u, v \in V$.
$\operatorname{In}(v)$ denotes the set of incoming links
of node~$v$,
$\operatorname{Out}(v)$ denotes the set of outgoing links of node~$v$.
Links without predecessors (whose begin nodes have no inputs)
are called \emph{entrances}.
Links without successors (whose end nodes have no outputs)
are called \emph{exits}.
Links other than entrances and exits are referred to as \emph{inner links}.
By $E^\mathrm{in}$ and $E^\mathrm{out}$ we denote
the set of entrances and the set of exits, and
assume that $E^\mathrm{in} \cap E^\mathrm{out} = \varnothing$.
Nodes correspond to road intersections, merges, diverges,
or subdivide longer links into smaller ones.

Every inner link or exit~$e$ has the following characteristics,
also known as a fundamental diagram:
\begin{center}
\begin{tabular}{ll}
$N_e$ & maximum number of vehicles in the link, \\
$F_e$ & capacity (in vehicles per time step), \\
$v_e$ & free flow speed, \\
$w_e$ & congestion wave speed.
\end{tabular}
\end{center}
For entrances only capacity and free flow speed need to be defined.
The free flow speed and the congestion speed are measured in link per time step,
$0 < v_e, w_e < 1$.

Link~$e$ at time~$t$ contains $n_e(t)$ vehicles.
Define the \emph{outflow demand}, or the required outflow,
of link~$e \in E$ at time~$t$ as
$f_e^d(t) = \min \{v_e n_e(t), F_e\}$.
Define the \emph{inflow supply}, or the maximum inflow,
of link~$e \in E \setminus E^\mathrm{in}$ at time~$t$ as
$f_e^s(t) = \min \{w_e (N_e - n_e(t)), F_e\}$.
The number of vehicles in link~$e = (u, v)$ at time~$t + 1$
is
\begin{equation} \label{eq:continuity}
n_e(t + 1) = n_e(t) + f_e^\mathrm{in}(t) - f_e^\mathrm{out}(t),
\end{equation}
where $f_e^\mathrm{in}(t)$ is the total incoming flow for link~$e$ at time~$t$:
\begin{equation}
f_e^\mathrm{in}(t) = \sum_{e_1 \in \operatorname{In}(u)} f_{e_1, e}(t),\qquad
e \in E \setminus E^\mathrm{in},
\end{equation}
the incoming flow $f_e^\mathrm{in}(t)$ for every entrance $e \in E^\mathrm{in}$
is given;
$f_e^\mathrm{out}(t)$ is the total outgoing flow for link~$e$ at time~$t$:
\begin{equation}
f_e^\mathrm{out}(t) = \sum_{e_2 \in \operatorname{Out}(v)} f_{e, e_2}(t),\qquad
e \in E \setminus E^\mathrm{out},
\end{equation}
the outgoing flow for exits $e \in E^\mathrm{out}$ equals the outflow demand:
$f_e^\mathrm{out}(t) = f_e^d(t)$.

Flows $f_{e_1, e_2}(t)$ between adjacent links $e_1$, $e_2$,
are determined by a node model
that should satisfy the following conditions:
(1)~all~flows between adjacent~nodes $f_{e_1, e_2}(t)$ are non-negative;
(2)~incoming~flows don't exceed inflow~supplies, except for entrances,
$f_e^\mathrm{in}(t) \le f_e^s(t)$, $e \in E \setminus E^\mathrm{in}$; and
(3)~outgoing~flows don't exceed outflow~demands,
$f_e^\mathrm{out}(t) \le f_e^d(t)$, $e \in E$.

Conservation law~\eqref{eq:continuity} and the flow constraints ensure that if
the number of vehicles at time~$t$ in each link~$n_e(t)$
is positive and less than the maximum number of vehicles~$N_e$
(that is, $n_e(t) \ge 0$ for all~$e \in E$
and $n_e(t) \le N_e$ for all $e \in E \setminus E^\mathrm{in}$),
then the same holds for the next time step $t + 1$
(namely, $n_e(t + 1) \ge 0$ for $e \in E$, and
$n_e(t + 1) \le N_e$ for $e \in E \setminus E^\mathrm{in}$).
Moreover, if an exit link~$e \in E^\mathrm{out}$,
is in a \emph{free flow state} at time~$t$, $v_e n_e(t) \le F_e$,
then it will stay in a free flow state at time~$t + 1$:
$v_e n_e(t + 1) \le F_e$.

\subsection{Node Model} \label{subsec:node-model}
We use the node model proposed in \cite{tampere11}.
Consider node~$v$ with $m$ incoming and $n$ outgoing links,
$m, n > 0$.
Let $f_i^d(t)$ be an outflow demand of the $i$th incoming link at time~$t$,
and $f_j^s(t)$ be an inflow supply of the $j$th outgoing link at time~$t$
as shown in Figure~\ref{fig:node}.

\begin{figure}[htb]
\centering
\includegraphics{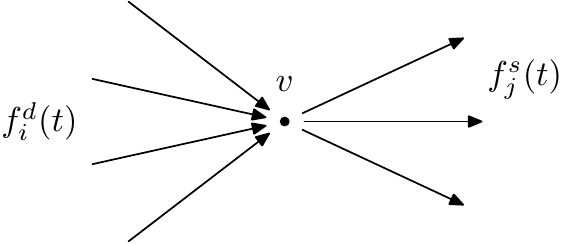}
\caption{Node of a road network graph.}
\label{fig:node}
\end{figure}

Let $f_{ij}(t)$ be the flow from the $i$th incoming to the $j$th outgoing link.
A node model determines flows $f_{ij}(t)$
from demands~$f_i^d(t)$, supplies~$f_j^s(t)$,
a \emph{split ratio matrix} $B_v(t) \in R_{m \times n}$
and \emph{priorities} of the incoming links $p_i \geq 0$.
\cite{tampere11} proposed that the priorities~$p_i$
should be proportional to capacities~$F_i$ of the incoming links.
Alternatively, the priorities~$p_i$
can be proportional to demands~$f_i^d(t)$,
as suggested in~\cite{jin03} and implemented in \cite{kurvar12}.

Elements of the split ratio matrix
$B_v(t) = \{ \beta_{ij}(t) \}_{i = 1, \dots, m}^{j = 1, \dots, n}$
are nonnegative, and
$\sum_{j = 1}^n \beta_{ij}(t) = 1$.
The split ratio matrix imposes the First-In-First-Out (FIFO) constraint
on the flows~$f_{ij}(t)$:
\begin{equation}
\frac{f_{i j_1}(t)}{\beta_{i j_1}(t)} =
\frac{f_{i j_2}(t)}{\beta_{i j_2}(t)}.
\end{equation}

Priority coefficients of the incoming links, $p_i$,
affect flows~$f_{ij}(t)$, if a supply of some outgoing link
is too small to accommodate a demand from the incoming links.
This will be clarified later.

We now describe the algorithm for finding flows~$f_{ij}$.
To simplify the notation, the time step~$t$ is omitted.

\paragraph{The Algorithm}
\begin{enumerate}
\item Compute oriented demands
$f_{ij}^d = f_i^d \beta_{ij}$
and directed priorities
$p_{ij} = p_i \beta_{ij}$,
$i = 1, \dots, m$; $j = 1, \dots, n$.

\item Define sets $J(1)$, $V_j(1)$, and supply residuals~$\tilde f_j^s(1)$,
$j \in J(1)$:
\begin{align*}
J(1) &= \{j\colon \textstyle\sum_{i = 1}^m f_{ij}^d > 0 \}; \\
V_j(1) &= \{i\colon f_{ij}^d > 0\}; \\
\tilde f_j^s(1) &= f_j^s.
\end{align*}
Clearly, $f_{ij} = 0$ if $j \notin J(1)$.

Sets $V_j(k)$ defined at each step~$k$ are the sets of incoming links~$i$
such that flows~$f_{ij}$ are still to be defined.
$J(k)$ is a set of outgoing links~$j$ such that $V_j(k) \ne \varnothing$.
The supply residual~$\tilde f_j^s(k)$ is the inflow supply of link~$j$
to be distributed among incoming links in $V_j(k)$.

\item Set $k \leftarrow 1$.

\item \label{item:loop-start}
If $J(k) = \varnothing$, stop.

\item \label{item:a-reduction}
For each $j \in J(k)$ compute a reduction factor
\begin{equation*}
a_j(k) = \frac{\tilde f_j^s(k)}{\sum_{i \in V_j(k)} p_{ij}}.
\end{equation*}

Determine the minimum of~$\{a_j(k)\}_{j \in J(k)}$:
$\hat a(k) = a_{\hat\jmath(k)}(k) = \min_{j \in J(k)} a_j(k)$.

\item Define
$U(k) = \{i \in V_{\hat\jmath(k)}(k)\colon f_i^d \le \hat a(k) p_i \}$.
It's a set of incoming links such that their outflow demand does not exceed
the rightful share in the inflow supply of the remaining outgoing links.
\begin{enumerate}
\item \label{item:demand-constrained}
If $U(k) \ne \varnothing$, then
for all~$i \in U(k)$, $j = 1, \dots, n$, set
$f_{ij} = f_{ij}^d$,
and compute
\begin{alignat*}{2}
V_j(k + 1) &= V_j(k) \setminus U(k), &\qquad& j \in J(k); \\
J(k + 1) &= \{j \in J(k)\colon V_j(k + 1) \ne \varnothing\}; \\
\tilde f_j^s(k + 1) &= \tilde f_j^s(k) - \textstyle\sum_{i \in U(k)} f_{ij}^d,
&& j \in J(k + 1).
\end{alignat*}

\item \label{item:supply-constrained}
If $U(k) = \varnothing$, then
for all $i \in V_{\hat\jmath(k)}(k)$, $j = 1, \dots, n$,
set
$f_{ij} = \hat a(k) p_{ij}$,
and compute
\begin{alignat*}{2}
V_j(k + 1) &= V_j(k) \setminus V_{\hat\jmath(k)}(k), &\qquad& j \in J(k); \\
J(k + 1) &= \{j \in J(k)\colon V_j(k + 1) \ne \varnothing\}; \\
\tilde f_j^s(k + 1) &= \tilde f_j^s(k)
- \textstyle\sum_{i \in V_{\hat\jmath(k)}(k)} \hat a(k) p_{ij},
&& j \in J(k + 1).
\end{alignat*}
Note that $\hat\jmath(k) \notin J(k + 1)$.
\end{enumerate}

\item Set $k \leftarrow k + 1$, and go to step~\ref{item:loop-start}.
\end{enumerate}

The first 3 steps are initialization steps.
As explained in~\cite{tampere11},
in each iteration we define the reduction factors~$a_j$
(step~\ref{item:a-reduction}),
which
the incoming flows
competing for the remaining supply of the $j$th outgoing link
should be set proportional to,
if they were not constrained by the demand of the corresponding incoming link,
so $f_{ij} = a_{j} p_{ij}$.
Then, a minimum of $a_j$, denoted by $\hat a$, whose index is $\hat\jmath$,
is selected
and imposed on all incoming links from $V_{\hat\jmath}$.
If some links are found to be demand-constrained
(case~\ref{item:demand-constrained}),
their outgoing flow is set equal to demand,
and those demand-constrained links~$U$ are removed from every~$V_j$.
Otherwise (case~\ref{item:supply-constrained})
the flows from $V_{\hat\jmath}$ are supply-constrained with factor~$\hat a$,
and all links from $V_{\hat\jmath}$ are removed from every~$V_j$.
The algorithm stops when all flows have been defined.
Since at least one incoming link is removed from all~$V_j$
in each iteration,
the algorithm finishes after at most $m$~iterations.

\subsection{Freeway Model}
The model of a freeway without toll lanes described below is
a combination of the Cell Transmission Model for networks, similar to \cite{daganzo95},
and the node model proposed in~\cite{tampere11} and explained above.

A freeway network consists of mainline links,
entrance links (mostly on-ramps), and exit links (mostly off-ramps).
Links are connected to each other by nodes.
Nodes correspond to junctions, merge and diverge points, and are also used
to split long links into shorter ones.
When a freeway network is built for simulation, nodes are typically placed so as to
keep link lengths in the range between $.1$ and $.2$ miles.

The characteristics of the $i$th mainline link are:
its capacity $F_i$, the maximum number of vehicles it can contain $N_i$,
a free flow speed~$v_i$ and a congestion speed~$w_i$.
The capacity is measured in vehicles per time step,
while both speeds are measured in link per time step, $0 < v_i, w_i < 1$.
The freeway contains $K$ links, numbered from~1 to~$K$,
the entrance link numbered~0, and the exit link numbered~$K + 1$.
Additionally,
there may be an on-ramp with capacity~$R_i$ and free flow speed~$v_i^r < 1$
attached to the beginning of mainline link~$i$,
and an off-ramp with capacity~$S_i$ at the end of the $i$th link,
$i = 1, \dots, K$.
If there is no on-ramp at the beginning of the $i$th link, then $R_i = 0$.

We now clarify how the model parameters, such as capacities and speeds,
are \emph{normalized},
that is, how
the speeds measured in miles per hour
are converted into the speeds measured in link per time step,
and how the capacities measured in vehicles per hour per lane
are converted into the capacities measured in vehicles per time step.
Suppose $L_i$ is the length of mainline link~$i$ measured in miles,
$\tau$ is the time step measured in hours,
$V_i$ is the free flow speed in link~$i$ measured in miles per hour,
$W_i$ is the congestion speed in link~$i$ measured in miles per hour,
$C_i$ is the capacity of link~$i$ measured in vehicles per hour per lane,
$k_i$ is the number of lanes in link~$i$.
Then $v_i = V_i \tau / L_i$, $w_i = W_i \tau / L_i$,
$F_i = C_i k_i \tau$.
The speeds $v^r_i$ and the capacities $R_i$, $S_i$ are normalized in the same way.
Hence the aforementioned inequalities $v_i, w_i < 1$
are equivalent to $\max\{ V_i, W_i \} \, \tau / L_i < 1$,
which is the Courant~--- Friedrichs~--- Lewy condition (see~\cite{cfl28}).
For example, if the minimum link length~$L_i$ is $0.2$~miles,
the maximum free flow speed~$v_i$ is 60~mph,
and the maximum congestion speed $w_i$ is 50~mph, then
the maximum time step~$\tau$ is $0.2 / 60$~hours, which is~12~seconds.

\begin{figure}[htb]
\centering
\includegraphics{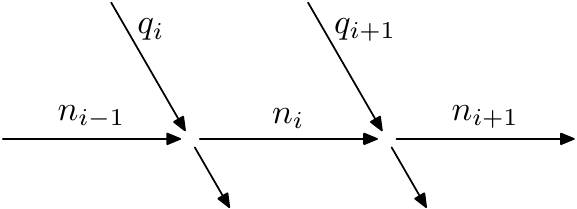}
\caption{Freeway model.}
\label{fig:freeway}
\end{figure}

Denote $n_i(t)$ the number of vehicles in link~$i$ at time~$t$.
We assume that if the $i$th link is uncongested,
that~is, $v_i n_i(t) \le F_i$,
then the number of vehicles it contains
doesn't influence the incoming flow for this link at time~$t$,
that is, $w_i (N_i - n_i(t)) \ge F_i$.
This assumption is equivalent to the following inequality:
\begin{equation} \label{eq:freeflow-noconstraint}
\frac{F_i}{v_i} + \frac{F_i}{w_i} \le N_i.
\end{equation}

Denote $d_i(t)$ the demand, or desired flow, at the $i$th on-ramp;
$r_i(t)$ --- the actual flow coming from on-ramp $i$ to link $i$;
and $q_i(t)$ --- the number of vehicles queued at the $i$th on-ramp,
at time $t$.
Both flows, $d_i(t)$ and $r_i(t)$ are measured in vehicles per time step.
Clearly,
\begin{equation} \label{eq:queue-dynamics}
q_i(t + 1) = q_i(t) + d_i(t) - r_i(t).
\end{equation}
Denote $f_i(t)$ the flow from link~$i$ to link~$(i + 1)$;
and $s_i(t)$ --- the flow from link~$i$ to the corresponding off-ramp.
The number of vehicles in
mainline links evolves as:
\begin{equation} \label{eq:mainline-dynamics}
n_i(t + 1) = n_i(t) + f_{i - 1}(t) + r_i(t) - f_i(t) - s_i(t).
\end{equation}
As for the entrance link~0 and the exit link~$K + 1$,
the same equality holds, except $r_i(t) = s_i(t) = 0$, $i = 0, K + 1$.

Since every exit remains in a free flow state once it's there,
there is no need to take the 
state of exits
into consideration:
assuming that all exits, including link~$K + 1$,
are initially in a free flow state,
and that inequality~\eqref{eq:freeflow-noconstraint} holds,
flows from mainline links to the exits are only constrained by the
exit capacities.

Off-ramp flow~$s_i(t)$ is proportional to the
flow to the downstream link~$f_i(t)$:
there exist split ratios $\beta_i^f$, $\beta_i^s$, such that
$\beta_i^f > 0$, $\beta_i^s \ge 0$, $\beta_i^f + \beta_i^s = 1$,
\begin{equation}
\frac{f_i(t)}{\beta_i^f} = \frac{s_i(t)}{\beta_i^s}.
\end{equation}
If link~$i$ has no off-ramp, let $\beta_i^s = 0$, $\beta_i^f = 1$.

Denote $F_i^s = F_i$, $i = 1, \dots, K + 1$,
$F_0^d = F_0$,
\begin{equation} \label{eq:outflow-capacity}
F_i^d = \begin{cases}
F_i, & \beta_i^s = 0, \\
\beta_i^f \min \{ F_i, S_i / \beta_i^s \}, & \beta_i^s > 0,
\end{cases}
\qquad
i = 1, \dots, K + 1.
\end{equation}
$F_i^s$ is the ``inflow capacity'' of the $i$th mainline link,
that is, the maximum incoming flow
if the link has enough space for them.
$F_i^d$ is the ``outflow capacity'' of the $i$th mainline link,
that is, the maximum number of vehicles which can move from this link
to link $i + 1$ per time step if link~$i + 1$ has enough space for them.
By introducing the outflow capacity
we ensure that the off-ramp flow~$s_i(t) = (\beta_i^s / \beta_i^f) f_i(t)$
never exceeds the off-ramp capacity~$S_i$
if the downstream flow~$f_i(t)$ does not exceed the outflow capacity~$F_i^d$:
\begin{equation}
s_i(t) = \frac{\beta_i^s}{\beta_i^f} f_i(t) \le
\frac{\beta_i^s}{\beta_i^f} F_i^d \le S_i.
\end{equation}

The incoming flows $d(t)$ and $f_{-1}(t)$ are given,
other flows are determined by the node model
(see Figure~\ref{fig:freeway-node}).
For each node compute outflow demands
for the on-ramp~$r_i^d(t)$ (if present)
and the incoming (upstream) mainline link~$f_{i - 1}^d(t)$,
and inflow supply for the outgoing (downstream) mainline link~$f_i^s(t)$:
\begin{alignat}{2}
f_i^s(t) &= \min \{ w_i (N_i - n_i(t)), F_i^s \}, &\qquad& i = 1, \dots, K + 1, \\
f_i^d(t) &= \min \{ \beta_i^f v_i n_i(t), F_i^d \}, && i = 0, \dots, K + 1, \\
r_i^d(t) &= \min \{ v_i^r q_i(t), R_i \}, && i = 1, \dots, K.
\end{alignat}
Let $r_{K + 1}^d(t) = 0$.
Note that $f_{K + 1}^s(t) = F_{K + 1}^s$.

\begin{figure}[htb]
\centering
\includegraphics{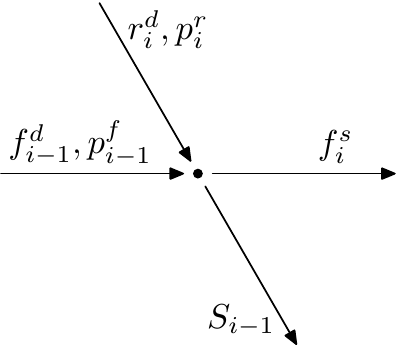}
\caption{Freeway node.}
\label{fig:freeway-node}
\end{figure}

The flow out of exit link~$K + 1$ 
equals the outflow demand,
$f_{K + 1}(t) = f_{K + 1}^d(t)$.
The flow from the last mainline link~$K$ to the exit link~$K + 1$
is a minimum of the outflow demand of the upstream link
and the inflow supply of the downstream link:
$f_K(t) = \min \{ f_K^d(t), f_{K + 1}^s(t) \}$.

Priorities of mainline links $p_i^f$, $i = 0, \dots, K - 1$,
and priorities of on-ramps $p_i^r$, $i = 1, \dots, K$, are given,
Suppose the priorities are normalized:
$p_{i - 1}^f + p_i^r = 1$.
For each $i = 1, \dots, K$ compute
$f_{i - 1}(t)$ and $r_i(t)$ as follows.
\begin{enumerate}
\item If the downstream mainline link can accomodate
the demand from the on-ramp and the upstream link,
$f_{i - 1}^d(t) + r_i^d(t) \le f_i^s(t)$, then
the actual flows equal the demands:
$f_{i - 1}(t) = f_{i - 1}^d(t)$,
$r_i(t) = r_i^d(t)$.
\item Otherwise, if the outflow demand from upstream mainline link
does not exceed its rightful share in the inflow supply of the downstream link,
$f_{i - 1}^d(t) \le p_{i - 1}^f f_i^s(t)$, then
the flow from the upstream link equals the demand,
$f_{i - 1}(t) = f_{i - 1}^d(t)$,
and the flow from the on-ramp takes up the residual of the supply:
$r_i(t) = f_i^s(t) - f_{i - 1}^d(t)$.
\item Otherwise, if the outflow demand from the on-ramp
doesn't exceed its rightful share in the supply,
$r_i^d(t) \le p_i^r f_i^s(t)$, then, analogously,
$r_i(t) = r_i^d(t)$,
$f_{i - 1}(t) = f_i^s(t) - r_i^d(t)$.
\item Otherwise
$f_{i - 1}^d(t) > p_{i - 1}^f f_i^s(t)$
and $r_i^d(t) > p_i^r f_i^s(t)$.
In this case the flows are proportional to the priorities:
$r_i(t) = p_i^r f_i^s(t)$,
$f_{i - 1}(t) = p_{i - 1}^f f_i^s(t)$.
\end{enumerate}

These rules are derived from the general node model presented
in subsection~\ref{subsec:node-model}.
Note that since $f_{i - 1}(t) \le f_{i - 1}^d(t) \le F_{i - 1}^d$,
the off-ramp flow $s_{i - 1}(t) = (\beta_{i - 1}^s / \beta_{i - 1}^f) f_{i - 1}(t)$ does not exceed the off-ramp capacity~$S_{i - 1}$
for the reasons mentioned after equation~\eqref{eq:outflow-capacity}.
Thus, both supply constraints~---
from the downstream link and from the off-ramp~---
are satisfied.

\subsection{Equilibria of the Freeway Model}
Suppose the incoming flows~$f_{-1}$ and~$d$ are constant,
as well as priorities~$p_i^f$, $p_i^r$
and split ratios $\beta_i^f$, $\beta_i^s$.
We say that a triple~$(n, f, r)$
is an \emph{equilibrium},
if
$v_{K + 1} n_{K + 1} \le F_{K + 1}$ (that is, the exit link $(K + 1)$ is uncongested)
and there exist queue lengths $n_0$, $q_i$, $i = 1, \dots, K$,
such that if $n_i(t) = n_i$, $i = 1, \dots, K$,
$n_0(t) = n_0$, and $q(t) = q$,
then $f(t + \Delta t) = f$, $n_i(t + \Delta t) = n_i$, $i = 1, \dots, K$,
and $r(t + \Delta t) = r$ for all $\Delta t = 0, 1, 2, \dots$\,.
Note that queue lengths $n_0$, $q$ are not included in the equilibrium.
Queue lengths either stay constant
or grow at a rate of $(f_{-1} - f_0)$ or $(d_i - r_i)$ vehicles per time step.

Define maximum flows $\bar f_0 = \min \{ f_{-1}, F_0 \}$,
$\bar r_i = \min \{ d_i, R_i \}$,
\begin{equation}
\bar f_i = \min \{ \beta_i^f (\bar f_{i - 1} + \bar r_i), F_i^d \},
\qquad
i = 1, \dots, K + 1.
\end{equation}
It can be shown that inequalities
$r_i \le \bar r_i$, $i = 1, \dots, K$, and
$f_i \le \bar f_i$, $i = 0, \dots, K + 1$,
hold for every equilibrium~$(n, f, r)$.

Theorems \ref{th:equilibrium-flows}, \ref{th:equilibrium-densities}
fully characterize
the set of all equilibria
corresponding to the incoming flows $f_{-1}$, $d$.

\begin{theo} \label{th:equilibrium-flows}
Equilibrium flows $r$ and $f$ are uniquely defined.
Namely, $f_{K + 1} = \bar f_{K + 1}$,
and flows $f_{i - 1}$,~$r_i$
are uniquely determined by~$f_i$, $i = K + 1, \dots, 1$,
as follows.
\begin{enumerate}
\item If $\bar f_{i - 1} \le p_{i - 1}^f f_i / \beta_i^f$, then
$f_{i - 1} = \bar f_{i - 1}$, $r_i = f_i / \beta_i^f - \bar f_{i - 1}$.
\item If $\bar r_i \le p_i^r f_i / \beta_i^f$, then
$r_i = \bar r_i$, $f_{i - 1} = f_i / \beta_i^f - \bar r_i$.
\item Otherwise $\bar f_{i - 1} > p_{i - 1}^f f_i / \beta_i^f$,
$\bar r_i > p_i^r f_i / \beta_i^f$; in this case
$f_{i - 1} = p_{i - 1}^f f_i / \beta_i^f$,
$r_i = p_i^r f_i / \beta_i^f$.
\end{enumerate}
\end{theo}
Note that this rule is similar to that in the node model.
Again, the off-ramp capacity constraints are not violated
due to the reasons mentioned after equation~\eqref{eq:outflow-capacity}
and the fact that $f \le \bar f \le F^d$.

Theorem~\ref{th:equilibrium-flows} follows from lemma~\ref{lemma:f-demand} and the node model.

\begin{lemma} \label{lemma:f-demand}
If $f_i < \bar f_i$, then $f_i < f_i^d$.
\end{lemma}
Lemma~\ref{lemma:f-demand} is purely technical, the proof is omitted.

Once the equilibrium flows $r$ and $f$ are known,
the set~$E$ of equilibrium density vectors~$n$
is defined as follows.
Let $I = \{ i\colon 1 \le i \le K,\ f_i = F_i^d \} \cup
\{ i\colon 1 \le i \le K,\ f_i + r_{i + 1} = F_{i + 1}^s \}$.
The set~$I$ is in fact a set of bottlenecks, since for all $i \notin I$
$f_i < F_i^d$ and $f_i < F_{i + 1}^s - r_{i + 1}$.
Suppose $I = \{ i_1, \dots, i_M \}$,
$i_1 < i_2 < \dots < i_M$, $M = |I|$.
Assume $i_0 = 0$.
Let
\begin{align}
I_m &= \{ i\colon i_{m - 1} < i \le i_m \},\qquad m = 1, \dots, M, \\
I_{M + 1} &= \{ i\colon i_M < i \le K + 1 \}.
\end{align}
If $I = \varnothing$, then
$I_{M + 1} = I_1 = \{1, \dots, K + 1\}$.
The set~$I_m$, $m \in \{2, \dots, M\}$, contains all links between
the bottleneck~$i_{m - 1}$
and the downstream bottleneck~$i_m$,
including the downstream bottleneck.
The set $I_{M + 1}$ contains links downstream of the~last bottleneck~$i_M$,
the set $I_1$ contains links upstream of the first bottleneck~$i_1$.

Let
\begin{equation}
n_i^u = n_i^u(r, f) = \frac{f_i}{\beta_i^f v_i},\qquad
n_i^c = n_i^c(r, f) = N_i - \frac{r_i + f_{i - 1}}{w_i},\qquad
i = 1, \dots, K + 1.
\end{equation}
Letters $u$ and $c$ here mean \emph{uncongested} and \emph{congested}.
We claim that $n^u \le n^c$ whenever $f_i = \beta_i^f (r_i + f_{i - 1}) \le F_i^d$.
Indeed, this follows from inequality~\eqref{eq:freeflow-noconstraint}:
\begin{equation}
n_i^u = \frac{f_i}{\beta_i^f v_i} \le \frac{F_i^d}{\beta_i^f v_i} \le
\frac{F_i}{v_i} \le N_i - \frac{F_i}{w_i}
\le N_i - \frac{F_i^d}{\beta_i^f w_i} \le
N_i - \frac{r_i + f_{i - 1}}{w_i} = n_i^c.
\end{equation}

Define sets
\begin{align}
U &= \{ i\colon 1 \le i \le K - 1,\ f_i < F_i^d,\ f_i / p_i^f < r_{i + 1} / p_{i + 1}^r \}, \\
C &= \{ i\colon 1 \le i \le K,\ r_i < \bar r_i,\ f_{i - 1} + r_i < F_i^s \}
\cup \{1, \text{ if } f_0 < \bar f_0,\ f_0 + r_1 < F_1^s \}.
\end{align}
$U$ is the set of uncongested links: $n_i = n_i^u$ for~$i \in U$;
$C$ is the set of congested links: $n_i = n_i^c$ for~$i \in C$.
These conditions can be obtained directly from the node model.
For~$m = 1, \dots, M$ define indices
\begin{equation}
i_m^u = \begin{cases}
i_{m - 1}, & I_m \cap U = \varnothing, \\
\max (I_m \cap U), & I_m \cap U \ne \varnothing,
\end{cases}\qquad
i_m^c = \begin{cases}
i_m + 1, & I_m \cap C = \varnothing, \\
\min (I_m \cap C), & I_m \cap C \ne \varnothing.
\end{cases}
\end{equation}
It can be shown that $i_m^u < i_m^c$ for each $m = 1, \dots, M$.

\begin{theo} \label{th:equilibrium-densities}
The set~$E$ of equilibrium vectors~$n$ is a direct product of sets
$E_m$ corresponding to sets~$I_m$:
\begin{equation}
E = \bigotimes_{m = 1}^{M + 1} E_m.
\end{equation}
The set~$E_{M + 1}$ consists of a single vector,
\begin{equation}
E_{M + 1} = \{ (n_{i_M + 1}^u, \dots, n_{K + 1}^u) \}.
\end{equation}
The set~$E_m$, $m \in \{1, \dots, M\}$ either consists
of a single vector,
\begin{equation}
E_m = \{ (n_{i_{m - 1} + 1}^u, \dots, n_{i_m^u}^u,
n_{i_m^c}^c, \dots, n_{i_m}^c) \},
\qquad
\text{if }i_m^u = i_m^c - 1,
\end{equation}
or is a union
\begin{equation}
E_m = \bigcup_{h = i_m^u + 1}^{i_m^c - 1} E_m^h,
\end{equation}
where
\begin{equation}
E_m^h = \{ (n_{i_{m - 1} + 1}^u, \dots, n_{h - 1}^u, n_h,
n_{h + 1}^c, \dots, n_{i_m}^c),\ n_h^u \le n_h \le n_h^c \}.
\end{equation}
\end{theo}
The proof of theorem~\ref{th:equilibrium-densities}
for $U = \varnothing$ and $C = \varnothing$
can be found in~\cite{gomes08}.

Figure~\ref{fig:equilibrium-densities} illustrates theorem~\ref{th:equilibrium-densities}.

\begin{figure}[htb]
\centering
\includegraphics{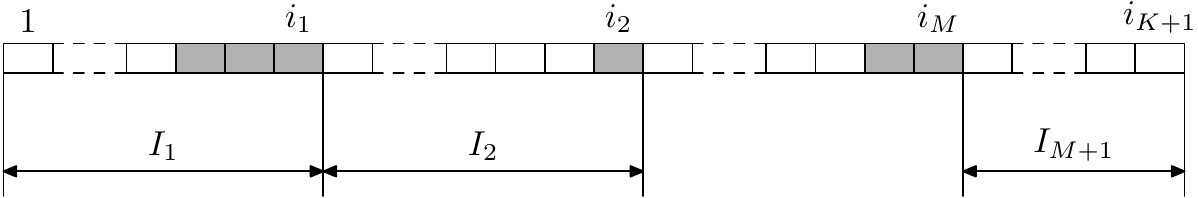}
\caption{Structure of equilibrium vector~$n$.}
\label{fig:equilibrium-densities}
\end{figure}

Note that the equilibrium flows and densities and thus the queue growth
depend on priorities $p^f$, $p^r$.
One extreme case, $p_i^r = 1$, $p_i^f = 0$ for all~$i$,
is fully studied by~\cite{gomes08}.
Another extreme case is $p_i^r = 0$, $p_i^f = 1$ for all~$i$.
In this case all non-bottleneck links are uncongested,
whereas the on-ramp just upstream of the last bottleneck 
is more congested than for any strictly positive on-ramp priorities~$p^r$.

\subsubsection{Feasible and infeasible incoming flows}
Flows $f_i$, $i = 1, \dots, K + 1$,
are uniquely defined by flows from entrances $f_0$ and $r_i$, $i = 1, \dots, K$,
because the equality $f_i = \beta_i^f (f_{i - 1} + r_i)$
holds for all $i = 1, \dots, K$.
Namely,
\begin{equation}
f_i(f_0, r) = f_0 \prod_{k = 1}^i \beta_k^f +
\sum_{j = 1}^i r_j \prod_{k = j}^i \beta_k^f,
\qquad
i = 1, \dots, K + 1.
\end{equation}
The incoming flow $(f_{-1}, d)$ is said to be \emph{feasible},
if $f_i(\bar f_0, \bar r) \le F_i^d$, $i = 1, \dots, K + 1$,
and \emph{infeasible} otherwise.
The incoming flow is said to be \emph{strictly feasible},
if $f_i(\bar f_0, \bar r) < F_i^d$, $i = 1, \dots, K + 1$.

It can be demonstrated that if the incoming flow~$(f_{-1}, d)$
is feasible, then the equilibrium flows equal maximum flows,
$r = \bar r$,
$f = \bar f = f(\bar f_0, \bar r)$,
and consequently $C = \varnothing$.
Moreover, if the incoming flow~$(f_{-1}, d)$
is strictly feasible, then
the set~$E$ of equilibrium density vectors
consists of a single vector, $E = \{ n^u \}$.

\subsection{Model of a Freeway with Toll Lanes}
We model a freeway with toll lanes by splitting
each mainline link into two parallel ones
(see Figure~\ref{fig:highway-managed-lanes}),
the first one (upper index~1)
corresponds to toll lanes,
the second one (upper index~2)
corresponds to the general purpose lanes.
Each of these two links has a separate exit.

\begin{figure}[htb]
\centering
\includegraphics{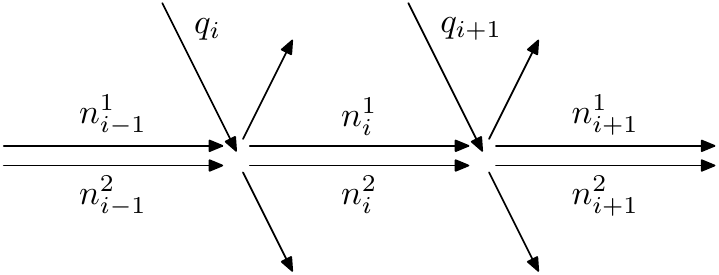}
\caption{Model of a freeway with toll lanes.}
\label{fig:highway-managed-lanes}
\end{figure}

The capacities and maximum number of vehicles
of the resulting links are proportional to the number of lanes.
Let $l_1$ be the number of toll lanes,
and $l_2$ be the number of general purpose lanes.
Then the capacities
$F_i^1 = F_i l_1 / (l_1 + l_2)$,
$F_i^2 = F_i l_2 / (l_1 + l_2)$.
The maximum number of vehicles $N_i^\xi$, 
the maximum flows $F_i^{\xi, d}$, 
$F_i^{\xi, s}$,
and priorities $p_i^{\xi, s}$,
$\xi = 1, 2$,
are defined in the same manner.

Outflow demands and inflow supplies are defined as in the freeway model:
\begin{alignat}{3}
r_i^d(t) &= \min \{ v_i^r q_i(t), R_i \}, &\qquad&&\quad& i = 1, \dots, K\\
f_i^{\xi, d}(t) &= \min \{ \beta_i^f v_i n_i^\xi(t), F_i^{\xi, d} \},
&& \xi = 1, 2,
&&  i = 0, \dots, K + 1, \\
f_i^{\xi, s}(t) &= \min \{ w_i (N_i^\xi - n_i^\xi(t)), F_i^{\xi, s} \},
&& \xi = 1, 2,
&& i = 1, \dots, K + 1.
\end{alignat}
Travelers can only choose between the toll and the general purpose lanes
when they enter the freeway.
The freeway state and the toll both influence
split ratios of the flow from the entrance into the mainline links.
Denote those split ratios by $\alpha_i^1$, $\alpha_i^2$.
Of course, $\alpha_i^1, \alpha_i^2 \ge 0$,
$\alpha_i^1 + \alpha_i^2 = 1$.

\begin{figure}[htb]
\centering
\includegraphics{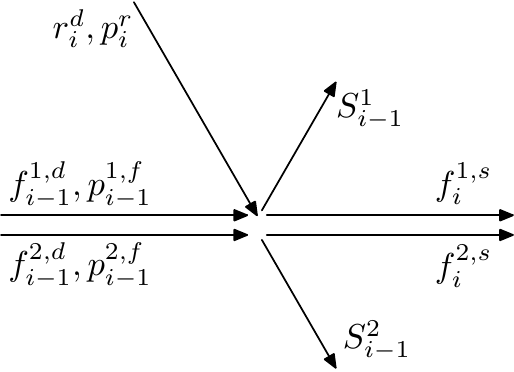}
\caption{Node model of a freeway with toll lanes.}
\label{fig:highway-ml-node}
\end{figure}

Flows $f_{-1}^\xi(t)$ are given,
flows out of the exit links equal demands:
$f_{K + 1}^\xi(t) = f_{K + 1}^{\xi, d}(t)$.
Other flows, namely $f_{i - 1}^\xi(t)$, $r_i^\xi(t)$, $i = 1, \dots, K + 1$,
are determined by the following node model (see Figure~\ref{fig:highway-ml-node}).
The time step~$t$ is implied, but omitted for simplicity.
Compute potential flows $\psi_i^1$, $\psi_i^2$
\begin{equation}
\psi_i^\xi = \min \left\{ \max \left\{
f_i^{\xi, s} \frac{\alpha_i^\xi p_i^r}{\alpha_i^\xi p_i^r + p_{i - 1}^{\xi, f}},
f_i^{\xi, s} - f_{i - 1}^{\xi, d}
\right\},
\alpha_i^\xi r_i^d \right\},
\qquad
\xi = 1, 2.
\end{equation}
The potential flow $\psi_i^1$
is the flow~$r_i^1$
from the entrance to the toll lane of the $i$th mainline link
(according to the general node model
from subsection~\ref{subsec:node-model}),
if this flow is not constrained by the general purpose lane.
Such a constraint may arise, as shown later,
since the flows $r_i^1$, $r_i^2$
should be proportional to the split ratios $\alpha_i^1$, $\alpha_i^2$.
Similarly, $\psi_i^2$ is the flow from the entrance
to the general purpose lane of the $i$th link,
if there are no constraints from the toll lane.

Define
\begin{equation}
\lambda_i^\xi = \begin{cases}
1, & \alpha_i^\xi = 0, \\
\psi_i^\xi / (\alpha_i^\xi r_i^d), & \alpha_i^\xi > 0,
\end{cases}
\qquad
\xi = 1, 2,
\end{equation}
and $\lambda_i = \min \{ \lambda_i^1, \lambda_i^2 \}$.
Clearly, $\lambda_i^\xi \le 1$, therefore $\lambda_i \le 1$.
Finally, compute flows $r_i^\xi$ and $f_{i - 1}^\xi$:
\begin{equation}
r_i^\xi = \lambda_i \alpha_i^\xi r_i^d,
\qquad
f_{i - 1}^\xi = \min \{ f_{i - 1}^{\xi, d}, f_i^{\xi, s} - r_i^\xi \},
\qquad
\xi = 1, 2.
\end{equation}

If the freeway is in free flow and the demand is feasible,
there is no need in toll lane, because it does not provide 
any gain in travel time.
On the other hand, if the freeway is in the fully congested
equilibrium, there is no logic
in redistributing the incoming flows~$r_i$ either:
even if the toll lane becomes uncongested,
the entrances will accumulate queues, reducing the
total output flow of the system and harming everyone.
So, in the case of fully congested freeway, even with feasible flow,
there is no use for a toll lane.
The situation when the toll lane is helpful though,
is when the freeway segment is partially congested.
That is, when without reducing the total output flow
of the system, toll lane can provide saving in travel
time for those choosing to use it.
Figure~\ref{fig:partially-congested} illustrates this case:
assume there is a partially congested freeway segment
(see also Figure~\ref{fig:equilibrium-densities}).
Once we divide this segment into the toll and the general purpose
lanes,
it is possible to redistribute vehicles between these two lanes 
so that the general purpose
lane acts as a storage for the extra vehicles, while the toll
lane frees up, and all of it retaining the original flows.

\begin{figure}[htb]
\centering
\includegraphics{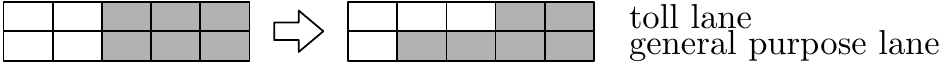}
\caption{Redistributing congestion of a partially congested highway.}
\label{fig:partially-congested}
\end{figure}

The proposed redistribution can be achieved
by controlling the split ratios~$\alpha_i^1$, $\alpha_i^2$
of the on-ramp flows~$r_i$.
This will be explained in the next section.


\section{Control of Split Ratios}\label{sec_control}
Here we present
an on-ramp flow redistribution control algorithm
which aims at bringing the toll lane into a free flow state
and maintaining the free flow state of the toll lane,
provided that queues at the entrances
grow as slowly as possible.

For simplicity, assume that the first link does not have an on-ramp,
and transform the model in such a way,
that the link~0 becomes an on-ramp of link~1.
Thus, there is no 0th link and $f_0^{\xi, s}(t) \equiv 0$, $\xi = 1, 2$.

We assume that
all information about the current freeway state,
namely, densities~$n_i^\xi(t)$ and queue lengths~$q_i(t)$,
is available.
But we know nothing about the future:
at time~$t$ no information about incoming flows
$d_i(t + \Delta t)$, $\Delta t = 0, 1, \dots$,
is available.

The pricing mechanism for using the toll lane
changes split ratios $\alpha_i^1(t)$, $\alpha_i^2(t)$.
This dependency will be discussed in section~\ref{sec_tolls}.
In this section we assume that
the split ratios $\alpha_i^1(t)$, $\alpha_i^2(t)$
can be set directly.

The algorithm for determining the split ratios
of the on-ramp flow~$r_i(t)$ at time step~$t$ is described next.
The time step~$t$ is implied, but will be omitted for simplicity.

\subsection{Queue Growth Rate Minimization}
First, we determine a range of split ratios that minimize
the queue growth rate.
It only makes sense to consider such links~$i$, where
$r_i^d > 0$ and $f_i^{1, s} + f_i^{2, s} > 0$,
since otherwise $r_i^1 = r_i^2 = 0$.

To minimize the queue growth rate,
$\lambda(\alpha_i^1, \alpha_i^2) = \min \{ \lambda_i^1(\alpha_i^1), \lambda_i^2(\alpha_i^2) \}$
should be maximized.
Denote
\begin{align}
\lambda_i^* &=
\max_{\alpha_i^1 \in [0, 1]} \lambda_i(\alpha_i^1, 1 - \alpha_i^1), \\
A_i^1 &= \argmax_{\alpha_i^1 \in [0, 1]} \lambda_i(\alpha_i^1, 1 - \alpha_i^1).
\end{align}
Our goal is to find~$A_i^1$.

It can be shown that $\lambda_i^\xi$ is
a monotonically decreasing function of~$\alpha_i^\xi$,
and $\lambda_i^\xi(\alpha_i^\xi)$ is continuous at least for $\alpha_i^\xi \in (0, 1]$,
$\xi = 1, 2$.
Indeed, if $\alpha_i^\xi \in (0, 1]$, then
\begin{equation}
\lambda_i^\xi(\alpha_i^\xi) = \min \left\{
\max \left\{
\frac{f_i^{\xi,s}}{r_i^d} \frac{p_i^r}{\alpha_i^\xi p_i^r + p_{i - 1}^{\xi, f}},
\frac{f_i^{\xi, s} - f_{i - 1}^{\xi, d}}{\alpha_i^\xi r_i^d}
\right\},
1 \right\}.
\end{equation}
Note that $\lambda_i^\xi(\alpha_i^\xi) \equiv 0$
for $\alpha_i^\xi \in (0, 1]$,
if $p_i^r = 0$ and $f_i^{\xi, s} \le f_{i - 1}^{\xi, d}$.

Let
\begin{equation}
\bar \alpha_i^\xi = \max \{
\alpha_i^\xi \colon
\alpha_i^\xi \in [0, 1],\ \lambda_i^\xi(\alpha_i^\xi) = 1 \}.
\end{equation}
Clearly, $\bar \alpha_i^\xi = \min \{ 1, \max \{ 0, a_i^\xi \} \}$,
where
\begin{equation}
a_i^\xi = \begin{dcases}
\frac{f_i^{\xi, s} - f_{i - 1}^{\xi, d}}{r_i^d},
& p_i^r = 0, \\
\max \left\{
\frac{f_i^{\xi, s} - f_{i - 1}^{\xi, d}}{r_i^d},
\frac{f_i^{\xi, s}}{r_i^d} - \frac{p_{i - 1}^{\xi, f}}{p_i^r}
\right\}, & p_i^r > 0.
\end{dcases}
\end{equation}
Consider the following cases.
\begin{enumerate}
\item $\bar \alpha_i^1 + \bar \alpha_i^2 \ge 1$.
In this case $\lambda_i^1(\alpha_i^1) = \lambda_i^2(\alpha_i^2) = 1$
and hence $\lambda_i(\alpha_i^1, \alpha_i^2) = 1$
for $\alpha_i^2 = 1 - \alpha_i^1$,
$\alpha_i^1 \in [1 - \bar \alpha_i^2, \bar \alpha_i^1]$.
Thus, $A_i^1 = [1 - \bar \alpha_i^2, \bar \alpha_i^1]$.

\item $\bar \alpha_i^1 + \bar \alpha_i^2 < 1$.
In this case $\lambda_i(\alpha_i^1, \alpha_i^2) < 1$
for all $\alpha_i^1, \alpha_i^2 \ge 0$, $\alpha_i^1 + \alpha_i^2 = 1$,
therefore, the flow from the on-ramp is inevitably constrained.

\begin{enumerate}
\item
If $p_i^r = 0$ and $f_i^{\xi, s} \le f_{i - 1}^{\xi, d}$, $\xi = 1, 2$,
then $\bar \alpha_i^1 = \bar \alpha_i^2 = 0$
and $\lambda(\alpha_i^1, 1 - \alpha_i^1) = 0$
for all $\alpha_i^1 \in [0, 1]$,
therefore $A_i^1 = [0, 1]$.

\item
Assume at least one of the inequalities
$p_i^r > 0$, $f_i^{\xi, s} > f_{i - 1}^{\xi, d}$, $\xi = 1, 2$, holds.

If $\lambda_i^1(1 - \bar \alpha_i^2) \ge \lim_{\alpha_2 \to \bar \alpha_i^2 + 0} \lambda_i^2(\alpha_2)$,
then
$A_i^1 = \{ 1 - \bar \alpha_i^2 \}$.

If $\lambda_i^2(1 - \bar \alpha_i^1) \ge \lim_{\alpha_1 \to \bar \alpha_i^1 + 0} \lambda_i^1(\alpha_1)$,
then
$A_i^1 = \{ \bar \alpha_i^1 \}$.

Otherwise there exists exactly one solution~$\alpha_i^{1, *}$
of the equation
$\lambda_i^1(\alpha_i^1) = \lambda_i^2(1 - \alpha_i^1)$,
and $\alpha_i^{1, *} \in (\bar \alpha_i^1, 1 - \bar \alpha_i^2)$,
since
$\lambda_i^\xi(\alpha_i^\xi)$ is strictly decreasing
for $\alpha_i^\xi \in (\bar \alpha_i^\xi, 1]$, $\xi = 1, 2$.
In this case $A_i^1 = \{ \alpha_i^{1, *} \}$.
\end{enumerate}
\end{enumerate}

So the set~$A_i^1$ is either a segment or a point.
The corresponding flow
\begin{equation}
r_i^1 = r_i^1(\alpha_i^1) =
\lambda_i(\alpha_i^1, 1 - \alpha_i^1) \alpha_i^1 r_i^d
\end{equation}
belongs to a segment
$[r_i^{1,\min}, r_i^{1,\max}]$,
where $r_i^{1,\min} = r_i^1(\min A_i^1)$,
$r_i^{1,\max} = r_i^1(\max A_i^1)$.

\subsection{On-Ramp Flow Redistribution}
In order not to redistribute on-ramp flows when
the freeway is in a free flow state,
introduce the following correction.
If $\bar \alpha_i^1 + \bar \alpha_i^2 > 1$
and
$\bar \alpha_i^1 > l_1 / (l_1 + l_2)$
(recall that $l_1$ and $l_2$ are the number of toll and general purpose lanes),
let
\begin{equation*}
\bar \alpha_i^1 = \max\{ 1 - \bar \alpha_i^2, l_1 / (l_1 + l_2) \}
\end{equation*}
and recalculate $r_i^{1, \max} = \bar \alpha_i^1 r_i^d$.

Compute the maximum maintainable level of densities
for the toll lane $n^{1,*}$.
A maintainable level of densities
is a vector~$n^{1, *}$
corresponding to a free-flow state
($\beta_i^f v_i n_i^{1, *} \le F_i^{1, d}$, $i = 1, \dots, K$),
such that
if $n^1(t) \le n^{1, *}$ and $r^1(t) = 0$,
then $n^1(t + 1) \le n^{1, *}$.
The maximum maintainable level of densities~$n^{1, *}$
is defined as follows.
First, compute flows~$f_i^{1, *}$, $i = K + 1, \dots, 1$ as follows:
$f_{K + 1}^{1, *} = F_{K + 1}^1$,
$f_i^{1, *} = \min \{ F_i^{1, d}, f_{i + 1}^{1, *} / \beta_{i + 1}^f \}$,
$i = K, \dots, 1$.
After that, compute densities~$n_i^{1, *} = f_i^{1, *} / (\beta_i^f v_i)$.

Next, define the maximum free-flow equilibrium~$n^{1, e}$.
Let $M$ be a number of on-ramps,
$i_1 < i_2 < \dots < i_M$ be the indices of links with on-ramps,
$i_1 = 1$ by assumption. Let $i_{M + 1} = K + 2$.
Define $f_{i_m}^{1, e} = f_{i_m}^{1, *}$,
$f_i^{1, e} = \beta_i^f f_{i - 1}^{1, e}$, $i = i_m + 1, \dots, i_{m + 1} - 1$,
$n_i^{1, e} = f_i^{1, e} / (\beta_i^f v_i)$.
Is is easily shown that $f^{1, e} \le f^{1, *}$
and $n^{1, e}$ is an equilibrium
corresponding to the on-ramp flow~$r^{1, e}$:
$r_{i_m}^{1, e} = f_{i_m}^{1, e} / \beta_{i_m}^f - f_{i_m - 1}^{1, e} \ge 0$.
Define
\begin{equation}
N^{1, e}(i_m, i_{m + 1}) = \sum_{i = i_m}^{i_{m + 1} - 1} n_i^{1,e}.
\end{equation}
The controller keeps the total number of vehicles
between entrances $m$ and $m + 1$ less than $N^{1, e}(i_m, i_{m + 1})$,
and the incoming flow $f_{i_m - 1}^1 + r_{i_m}^1$
less than equilibrium flow $f_{i_m}^{1, e} / \beta_{i_m}^f$,
if the minimum-queue-growth-rate constraint permits that.

If we knew flows $r_{i_m}^1$, deriving
split ratios $\alpha_{i_m}^1$, $\alpha_{i_m}^2$ would be easy:
if $r_{i_m}^d = 0$ or~$\lambda_{i_m}^* = 0$,
then the on-ramp flow $r_{i_m} = 0$ regardless of split ratios,
otherwise
$\alpha_{i_m}^1 = r_{i_m}^1 / r_{i_m}$, $\alpha_{i_m}^2 = 1 - \alpha_{i_m}^1$,
where $r_{i_m} = \lambda_{i_m}^* r_{i_m}^d$.
So, we need to determine
on-ramp flows $r_{i_m}^1 \in [r_{i_m}^{1,\mathrm{min}}, r_{i_m}^{1, \mathrm{max}}]$.

At every time step~$t$, compute the total number of vehicles
between every pair of adjacent on-ramps
$n^1(t, i_m, i_{m + 1}) = \sum_{i = i_m}^{i_{m + 1} - 1} n_i^1(t)$,
estimate flows $s_i^1(t)$, $i = 1, \dots, K$,
and $f_{i_m - 1}^1(t)$, $m = 1, \dots, M + 1$,
for $r_{i_m}^1(t) = r_{i_m}^{1, \mathrm{min}}(t)$.
By $\tilde s_i^1(t)$ and $\tilde f_{i_m - 1}^1(t)$ denote the estimates.

If
$f_{i_m - 1}^{1, d}(t) + r_{i_m}^{1, \mathrm{max}}(t) > \min \{ f_{i_m}^{1, s}(t), f_{i_m}^{1, e} / \beta_{i_m}^f \}$,
then
to possibly avoid constraining the flow from the upstream link~$f_{i_m - 1}$
and exceeding the equilibrium inflow~$f_{i_m}^{1, e} / \beta_{i_m}^f$,
set
\begin{equation*}
r_{i_m}^{1, \mathrm{max}}(t) =
\max \{ r_{i_m}^{1, \mathrm{min}}(t),
\min \{ f_{i_m}^{1, s}(t), f_{i_m}^{1, e} / \beta_{i_m}^f \} - f_{i_m - 1}^{1, d}(t) \}.
\end{equation*}

\begin{enumerate}
\item Let $\Delta n = 0$, $m = M$, $\gamma_M = 1$.

Here $\Delta n$ is an excess of vehicles in the toll lane,
$m$ is the number of current entrance, from~$M$ down to~1,
$0 \le \gamma_m \le 1$ is a reduction coefficient.
If $\gamma_m$ is strictly less than~1,
it means that some downstream entrance is unable to reduce congestion
in the toll lane on its own.
Therefore, we reduce the target number of vehicles between
entrances $m$ and $m + 1$ by $\gamma_m$
in order to reduce the flow to the downstream links.

\item \label{item:control-loop}
Compute
$\Delta n^1(t, i_m, i_{m + 1}; \gamma_{i_m}) =%
n^1(t, i_m, i_{m + 1}) - \gamma_m N^{1, e}(i_m, i_{m + 1})$,
\begin{equation*}
\Delta n \leftarrow \Delta n + \Delta n^1(t, i_m, i_{m + 1}; \gamma_m)
+ \tilde f_{i_m - 1}^1(t) - \tilde f_{i_{m + 1} - 1}^1(t)
- \sum_{i = i_m}^{i_{m + 1} - 1} \tilde s_i(t).
\end{equation*}
We estimate the excess number of vehicles in the toll lane
in links $i_m, \dots, i_{m + 1} - 1$.

\item Determine
$r_{i_m}^1(t) = \max \{ r_{i_m}^{1, \mathrm{min}}(t),%
\min \{ r_{i_m}^{1, \mathrm{max}}(t), - \Delta n \} \}$.

The on-ramp flow $r_{i_m}^1$ should reduce the excess number of vehicles, if possible.

\item If $m = 1$, stop: all on-ramp flows for
the current time step have been determined.

\item Recalculate $\Delta n \leftarrow \max\{ 0, \Delta n + r_{i_m}^1(t) \}$.

\item Compute the reduction coefficient for the upstream entrance:
$\gamma_{m - 1} = \min\{ 1, (f_{i_m}^{1, s}(t) - r_{i_m}^1(t)) / f_{i_m - 1}^{1, e} \}$.

\item Set $m \leftarrow m - 1$ and go to step~\ref{item:control-loop}.
\end{enumerate}

\subsection{Theorems}
We now prove that the presented controller
keeps the toll lane in a free-flow state,
if the minimum-queue-growth-rate condition allows it.
Moreover, even if the toll lane is initially congested,
it becomes almost uncongested in finite time.

The toll lane is considered as an independent system,
and the upper index $\xi = 1$ is omitted for simplicity.

\begin{theo} \label{th:freeflow-preserve}
Suppose $n(t) \le n^e$ and the inequality
$r_{i_m}(t) + f_{i_m - 1}^d(t) \le f_{i_m}^e / \beta_{i_m}^f$
holds for all entrances~$i_m$.
Then $n(t + 1) \le n^e$.
\end{theo}
\begin{proof}
Note that for $m > 1$, $n_{i_m - 1}(t) \le n^e_{i_m - 1}$ and
\begin{equation*}
f_{i_m - 1}^d(t) \le f_{i_m - 1}^e \le f_{i_m}^e / \beta_{i_m}^f,
\end{equation*}
therefore the inequality
$r_{i_m}(t) + f_{i_m - 1}^d(t) \le f_{i_m}^e / \beta_{i_m}^f$
holds at least for $r_{i_m}(t) = 0$.

Clearly, $f_{i - 1}(t) = f_{i - 1}^d(t)$ for $i = 2, \dots, K + 1$, since
$f_i^s(t) = F_i^s \ge f_i^e / \beta_i^f \ge f_{i - 1}^e \ge f_{i - 1}^d(t)$
if there is no~on-ramp in link~$i$,
and
$f_i^s(t) - r_i(t) = F_i^s - r_i(t) \ge f_i^e / \beta_i^f - r_i(t) \ge f_{i - 1}^d(t)$,
otherwise.
Therefore,
$f_i(t) + s_i(t) = f_i^d(t) / \beta_i^f = v_i n_i(t)$
for all~$i = 1, \dots, K + 1$.
Additionally,
$f_{i - 1}(t) \le f_{i - 1}^e \le f_i^e / \beta_i^f$
if there is no on-ramp in link~$i$,
and $r_i(t) + f_{i - 1}(t) \le f_i^e / \beta_i^f$, otherwise.
Thus, if $n(t) \le n^e$, then
\begin{equation*}
n_i(t + 1) = n_i(t) + f_{i - 1}(t) + r_i(t) - f_i(t) - s_i(t)
\le n_i(t) (1 - v_i) + f_i^e / \beta_i^f
\le n_i^e (1 - v_i) + f_i^e / \beta_i^f
= n_i^e
\end{equation*}
for all~$i$.
\end{proof}
This theorem means that if the state of the toll lane is in a ``target zone''
$n(t) \le n^e$ at time~$t$, then it remains there at time~$t + 1$
if the on-ramp flow demand~$r^d(t)$ isn't too high
and the general purpose lane isn't too congested near the entrances.

\begin{lemma} \label{lemma:monotonicity}
Suppose $n^-(t) \le n^+(t)$ and $r^-(t) \le r^+(t)$.
Then $n^-(t + 1) \le n^+(t + 1)$.
\end{lemma}
A similar lemma is proved in~\cite{gomes08}.

\begin{theo} \label{th:freeflow-asymptotic}
Suppose $f^e$ is an equilibrium flow,
the inequalities
\begin{equation} \label{ineq:th2-flows}
f_{i_m - 1}(t) + r_{i_m}(t) \le f_{i_m}^e / \beta_{i_m}^f,
\qquad
f_{i_m}^s(t) \ge f_{i_m - 1}^e
\end{equation}
hold for all links with on-ramps~$i_m$ for all~$t$,
and
\begin{equation} \label{ineq:th2-densities}
\sum_{i = i_m}^{i_{m + 1} - 1} n_i(t) \le \sum_{i = i_m}^{i_{m + 1} - 1} n_i^e,
\qquad
m = 1, \dots, M,
\end{equation}
for all~$t$.
Here $n_i^e = f_i^e / (\beta_i^f v_i)$.

Then for any $\varepsilon > 0$ there exists
a time step $T = T(\varepsilon)$ such that
$n_i(t) \le n_i^e + \varepsilon$, $i = 1, \dots, K + 1$,
for $t \ge T(\varepsilon)$.
\end{theo}
\begin{proof}[Proof of Theorem~\ref{th:freeflow-asymptotic} (Outline)]
First transform the system to get rid of the off-ramps.
Multiply the ``inflow capacity'' of the mainline link $F_i^s$,
the on-ramp capacity $R_i$,
the maximum number of vehicles~$N_i$,
the incoming flow~$d_i(t)$,
the number of vehicles in mainline link $n_i(t)$ and on the on-ramp $q_i(t)$
by $\mu_i = \prod_{j = 1}^{i - 1} (\beta_j^f)^{-1}$.
Multiply the ``outflow capacity'' of the $i$th link $F_i^d$ by $\mu_{i + 1}$.
For the transformed system denoted by~$\hat{\phantom m}$
the following equations hold:
\begin{align*}
\hat f_{i - 1}^d(t)
&= \min \{ v_{i - 1} \hat n_{i - 1}(t), \hat F_{i - 1}^d \}
= \mu_i f_{i - 1}^d(t), \\
\hat f_i^s(t)
&= \min \{ w_i (\hat N_i - \hat n_i(t)), \hat F_i^s \}
= \mu_i f_i^s(t), \\
\hat r_i^d(t) &= \min \{ v_i^r \hat q_i(t), \hat R_i \}
= \mu_i r_i^d(t).
\end{align*}
Therefore,
$\hat f_{i - 1}(t) = \mu_i f_{i - 1}(t)$,
$\hat r_i(t) = \mu_i r_i(t)$,
$\hat f_i(t) = \mu_{i + 1} f_i(t) = \mu_i f_i(t) / \beta_i^f$,
thus $\hat n_i(t + 1) = \mu_i n_i(t + 1)$,
$\hat q_i(t + 1) = \mu_i q_i(t + 1)$.
That is, the transformed system is equivalent to the initial system.
Multiply the equilibrium flow $f_i^e$ by $\mu_{i + 1}$.
The inequalities
$f_{i_m - 1}(t) + r_{i_m}(t) \le f_{i_m}^e / \beta_{i_m}^f$,
$f_{i_m}^s(t) \ge f_{i_m - 1}^e$
transform into
$\hat f_{i_m - 1}(t) + \hat r_{i_m}(t) \le \hat f_{i_m}^e$,
$\hat f_{i_m}^s(t) \ge \hat f_{i_m - 1}^e$,
while the inequality
$\sum_{i = i_m}^{i_{m + 1} - 1} n_i(t) \le \sum_{i = i_m}^{i_{m + 1} - 1} n_i^e$
transforms into
$\sum_{i = i_m}^{i_{m + 1} - 1} \mu_i^{-1} \hat n_i(t) %
\le \sum_{i = i_m}^{i_{m + 1} - 1} \mu_i^{-1} \hat n_i^e$.
Hence, it suffices to show that the theorem holds for
the freeway with no off-ramps, but the inequality
$\sum_{i = i_m}^{i_{m + 1} - 1} n_i(t) \le \sum_{i = i_m}^{i_{m + 1} - 1} n_i^e$
should be replaced with
\begin{equation*}
\sum_{i = i_m}^{i_{m + 1} - 1} \alpha_i n_i(t)
\le \sum_{i = i_m}^{i_{m + 1} - 1} \alpha_i n_i^e,
\end{equation*}
where
$\alpha_i = \mu_i^{-1} = \prod_{j = 1}^{i - 1} \beta_j^f$.
Note that
$1 = \alpha_1 \ge \alpha_2 \ge \dots \ge \alpha_{K + 1} > 0$.

Consider the links between two adjacent on-ramps.
Consider the clusters of links with $n_i(t) > n_i^e$.
It is easy to show that the downstream boundary of a cluster
can only move downstream,
one cluster cannot split into two or more,
and new clusters cannot be created.
So, starting from some point in time,
the number of clusters remains constant
and their downstream boundaries are fixed.

Denote
\begin{equation*}
h_m(t) = \sum_{i = i_m}^{i_{m + 1} - 1} \max \{ 0, n_i(t) - n_i^e \}.
\end{equation*}
We shall prove that $h_m(t) \to 0$ as $t \to \infty$ for all~$m$.

If there are no clusters between on-ramps $m$ and $m + 1$,
then, clearly, $h_m = 0$.

It is easily shown that $h_m(t)$ is decreasing and nonnegative,
thus there exists $\lim_{t \to \infty} h_m(t) \ge 0$.
Suppose~$\lim_{t \to \infty} h_m(t)  = \delta > 0$.
Let $i_m^*$ be the downstream boundary of the most upstream cluster
between the on-ramps $m$ and $m + 1$.
Using lemma~\ref{lemma:monotonicity},
it can be demonstrated that $\liminf_{t \to \infty} n_i(t) \ge n_i^e$,
$i = i_m^* + 1, \dots, i_{m + 1} - 1$.
Consider only those time steps where the number of clusters
and their downstream boundaries are stabilised (that is, do not change anymore)
and
$n_i^e - \eta_i \le n_i(t)$ for some small $\eta_i > 0$,
say, $\eta_i = \delta \alpha_{i_{m + 1}} / (2 (i_{m + 1} - i_m) \alpha_i)$,
$i = i_m^* + 1, \dots, i_{m + 1} - 1$.
We now demonstrate that the most upstream cluster disappears in finite time.
Since $h_m(t) \ge \delta$,
\begin{align*}
0 &\ge \sum_{i = i_m}^{i_m - 1} \alpha_i (n_i(t) - n_i^e) = \\
&= \sum_{\substack{i < i_m^*, \\ n_i(t) \le n_i^e}} \alpha_i (n_i(t) - n_i^e)
+ \sum_{n_i(t) > n_i^e} \alpha_i (n_i(t) - n_i^e)
+ \sum_{\substack{i > i_m^*, \\ n_i(t) \le n_i^e}} \alpha_i (n_i(t) - n_i^e) \ge \\
&\ge \alpha_{i_m} \sum_{\substack{i < i_m^*, \\ n_i(t) \le n_i^e}} (n_i(t) - n_i^e)
+ \delta - \sum_{i = i_m}^{i_{m + 1} - 1} \alpha_i \eta_i
= \alpha_{i_m} \sum_{\substack{i < i_m^*, \\ n_i(t) \le n_i^e}} (n_i(t) - n_i^e)
+ \alpha_{i_{m + 1}} \frac{\delta}2.
\end{align*}
Thus the following inequality holds:
\begin{equation} \label{eq:before-cluster}
\sum_{\substack{i < i_m^*, \\ n_i(t) \le n_i^e}} (n_i^e - n_i(t))
\ge \frac{\delta \alpha_{i_{m + 1}}}{2\alpha_{i_m}}.
\end{equation}
Consequently, at each time step~$t$ there exists $i \in \{i_m, \dots, i_m^* - 1\}$
such that
\begin{equation*}
n_i(t) \le n_i^e - \frac{\delta \alpha_{i_{m + 1}}}{2 \alpha_{i_m} (i_m^* - i_m)}.
\end{equation*}
This implies that oftener than every $(i_m^* - t_m)$ steps
the upstream cluster size will decrease by at least
\begin{equation*}
\beta_{i_m}^f \times \ldots \times \beta_{i_m^* - 1}^f \times
\frac{\delta \alpha_{i_{m + 1}}}{2 \alpha_{i_m} (i_m^* - i_m)}
= \mu > 0.
\end{equation*}
Therefore, the upstream cluster exists at most
$\lceil h_m(t) (i_m^* - i_m) / \mu\rceil$ time steps.
This contradicts our assumption that the number of clusters is stabilized.
Thus $h_m(t) \to 0$ as $t \to 0$,
which concludes the proof.
\end{proof}

\begin{remark}
In theorem~\ref{th:freeflow-asymptotic}
it would be sufficient to require that the inequality~\ref{ineq:th2-densities}
\begin{equation*}
\sum_{i = i_m}^{i_{m + 1} - 1} n_i(t) \le %
\sum_{i = i_m}^{i_{m + 1} - 1} n_i^e
\end{equation*}
holds only asymptotically, that is,
\begin{equation*}
\limsup_{t \to +\infty} \sum_{i = i_m}^{i_{m + 1} - 1} n_i(t) \le
\sum_{i = i_m}^{i_{m + 1} - 1} n_i^e.
\end{equation*}
\end{remark}

The control algorithm attemps to satisfy
the inequalities of theorem~\ref{th:freeflow-asymptotic},
$\sum_{i = i_m}^{i_{m + 1} - 1} n_i(t + 1) \le \sum_{i = i_m}^{i_{m + 1} - 1} n_i^e$
and
$f_{i_m - 1}(t) + r_{i_m}(t) \le f_{i_m}^e / \beta_{i_m}^f$,
at step~3.
If that is impossible,
a reduction coefficient~$\gamma_{m - 1}$ is introduced at step~6
to reduce the flow from upstream links.
If the requirements if the theorem can be met,
then the toll lane becomes almost uncongested
in finite time,
that is, for arbitrarily small~$\varepsilon > 0$
the density vector~$n^1(t)$ is component-wise
less than $n^{1, e} + \varepsilon$,
starting from some time step~$t$,
regardless of the initial state.

\begin{remark}
The suggested controller calculates only split ratios for on-ramp flow~$r$.
However, this scheme can be modified so, that
mainline travelers are also allowed to change lanes,
paying a toll if they change from the general purpose to the toll lane.
Since the controller keeps the toll lane less congested than general purpose lane,
only the travelers from the general purpose lane may want to change lanes.
In this model the split ratios of on-ramp flow~$r$
and the split ratios of flow~$f^2$ from general purpose lane
both depend on the price.
The price for entering the toll lane should be the same
for those drivers who come from on-ramps and
for those coming from the general purpose lane.
For that reason the split ratios for on-ramp flow~$r_i$ and the split ratios for
the flow from the general purpose lane~$f_{i - 1}^2$ are ``paired'',
since they are both defined by the same price.
The actual flows are defined by the generic node model presented in~Subsection~\ref{subsec:node-model}.

Some constraints, such as nonnegativity and boundedness,
can be imposed on tolls.
A toll is \emph{admissible} if it satisfies all constraints.
To narrow the set of split ratios corresponding to admissible tolls
we could introduce additional requirements,
for example, the minimization of queue length growth
or the maximization of total flow through the node.
Then an algorithm similar to the one presented above can be applied
at each time step.
\end{remark}

\subsection{Examples}

Consider a freeway with one on-ramp (in the first link)
and no exits, except for the last link, $(K + 1)$.
All links have equal capacities, free flow speeds and congestion speeds,
but the last link, $K + 1$, is a bottleneck with a lower capacity:
$F_{K + 1}^\xi < F_K^\xi = F_{K - 1}^\xi = \dots = F_1^\xi$, $\xi = 1, 2$.
The initial state of the freeway is an equilibrium.

\begin{figure}[htb]
\centering
\includegraphics{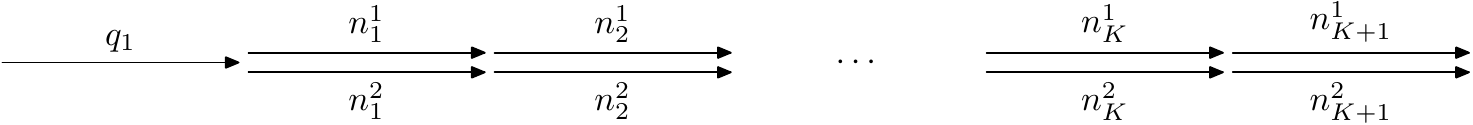}
\caption{Freeway with one entrance and one exit.}
\label{fig:freeway-1-1}
\end{figure}

\paragraph{Scenario~1}
Both the toll lane and the general-purpose lanes are uncongested
from link~1 to some link~$i$ and
congested from link~$(i + 1)$ to link~$(K + 1)$.
The incoming flow equals the bottleneck capacity: $d_0 = F_{K + 1}^1 + F_{K + 1}^2$.
The on-ramp flow~$r_0$ is redistributed in order to reduce congestion in the toll lane,
but only until there are uncongested links in the general purpose lane.

\emph{Scenario~1a}:
The number of uncongested links is too small
for the toll lane to become fully uncongested:
$i < (K + 1) / 2$.
Figure~\ref{fig:scenario-1a} illustrates this case,
the densities (in vehicles per mile) are color-coded.

\def\scale{.75}

\begin{figure}[htbp]
\centering
\includegraphics[scale=\scale]{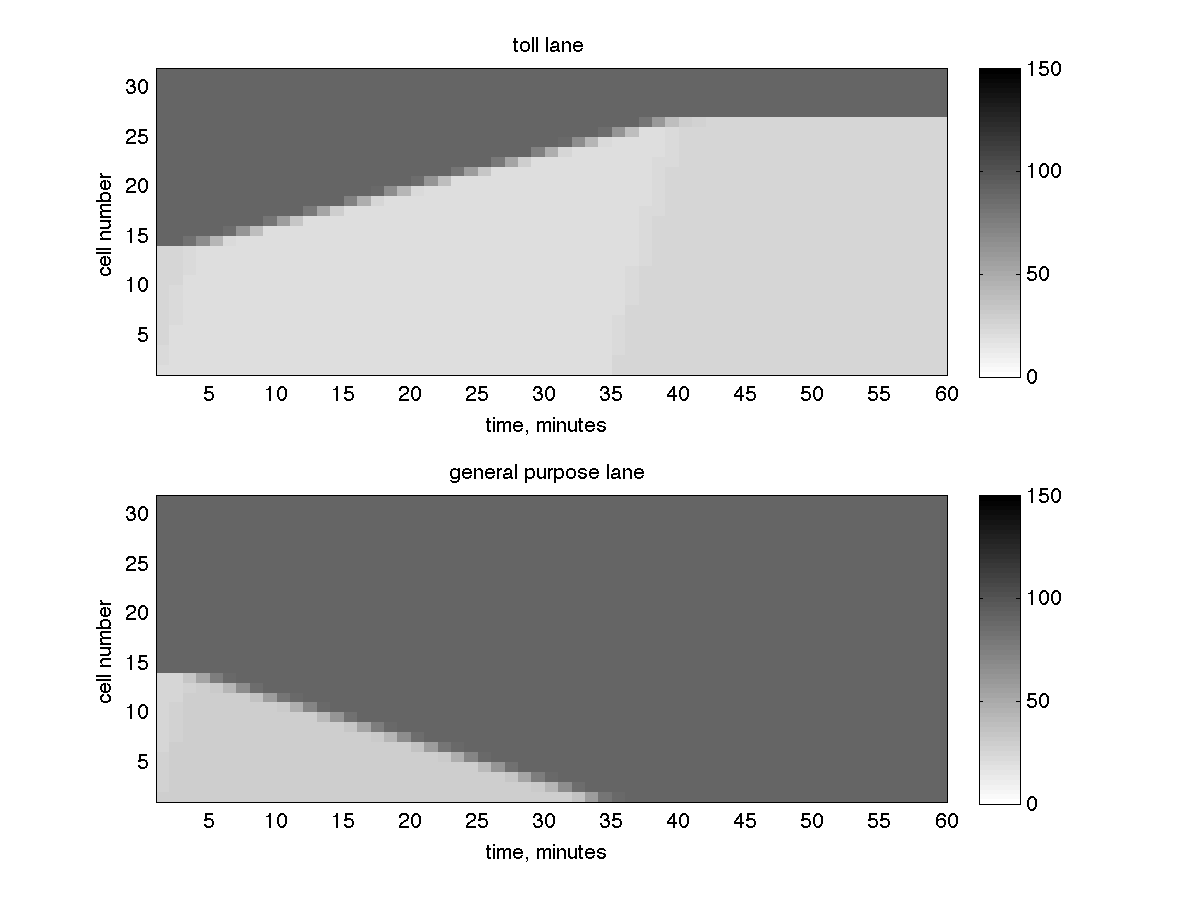}
\caption{Reducing congestion in the toll lane. Scenario 1a. Densities (in vehicles per mile).}
\label{fig:scenario-1a}
\end{figure}

\emph{Scenario~1b}:
The number of uncongested links is sufficiently large,
$i > (K + 1) / 2$, and the toll lane becomes fully uncongested in finite time
(Figure~\ref{fig:scenario-1b}).

\begin{figure}[htbp]
\centering
\includegraphics[scale=\scale]{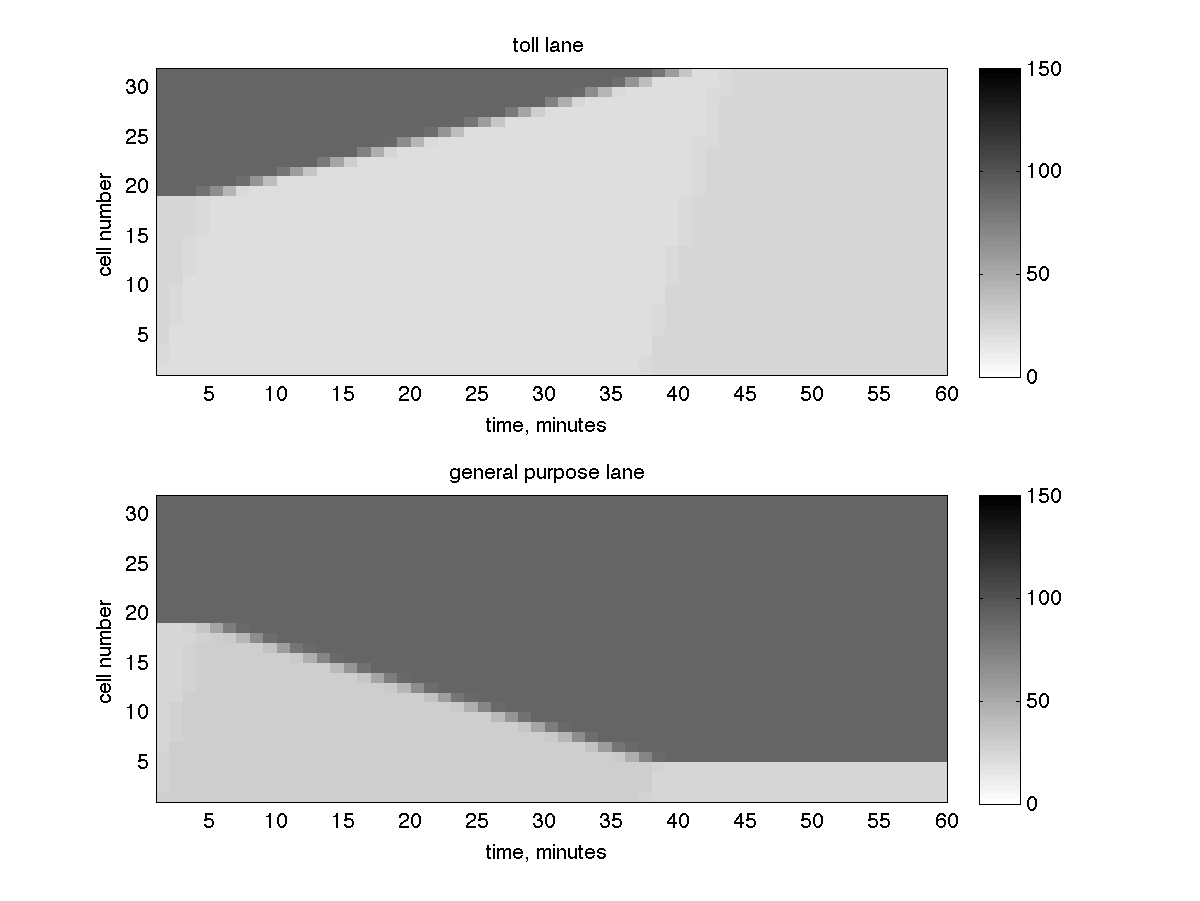}
\caption{Reducing congestion in the toll lane. Scenario 1b. Densities (in vehicles per mile).}
\label{fig:scenario-1b}
\end{figure}

\paragraph{Scenario~2}
Initially, the whole freeway is uncongested,
the incoming flow~$d_0$ equals the bottleneck capacity~$F_{K + 1}^1 + F_{K + 1}^2$.
At some point in time ($t = 5$~minutes), the incoming flow~$d_0$
exceeds the bottleneck capacity and
both the general purpose and the toll lanes become congested,
but the congestion grows faster in the general purpose lane.
The toll lane still becomes partially congested,
because the general purpose lane alone cannot accommodate
the whole excess flow: $d_0 > F_{K + 1}^1 + F_1^2$.
Then (at time $t = 30$~minutes)
the incoming flow~$d_0$ drops to the bottleneck capacity
and the congestion is transferred from the toll lane
to the general purpose lane.
Figure~\ref{fig:scenario-2} illustrates this scenario.

\begin{figure}[htbp]
\centering
\includegraphics[scale=\scale]{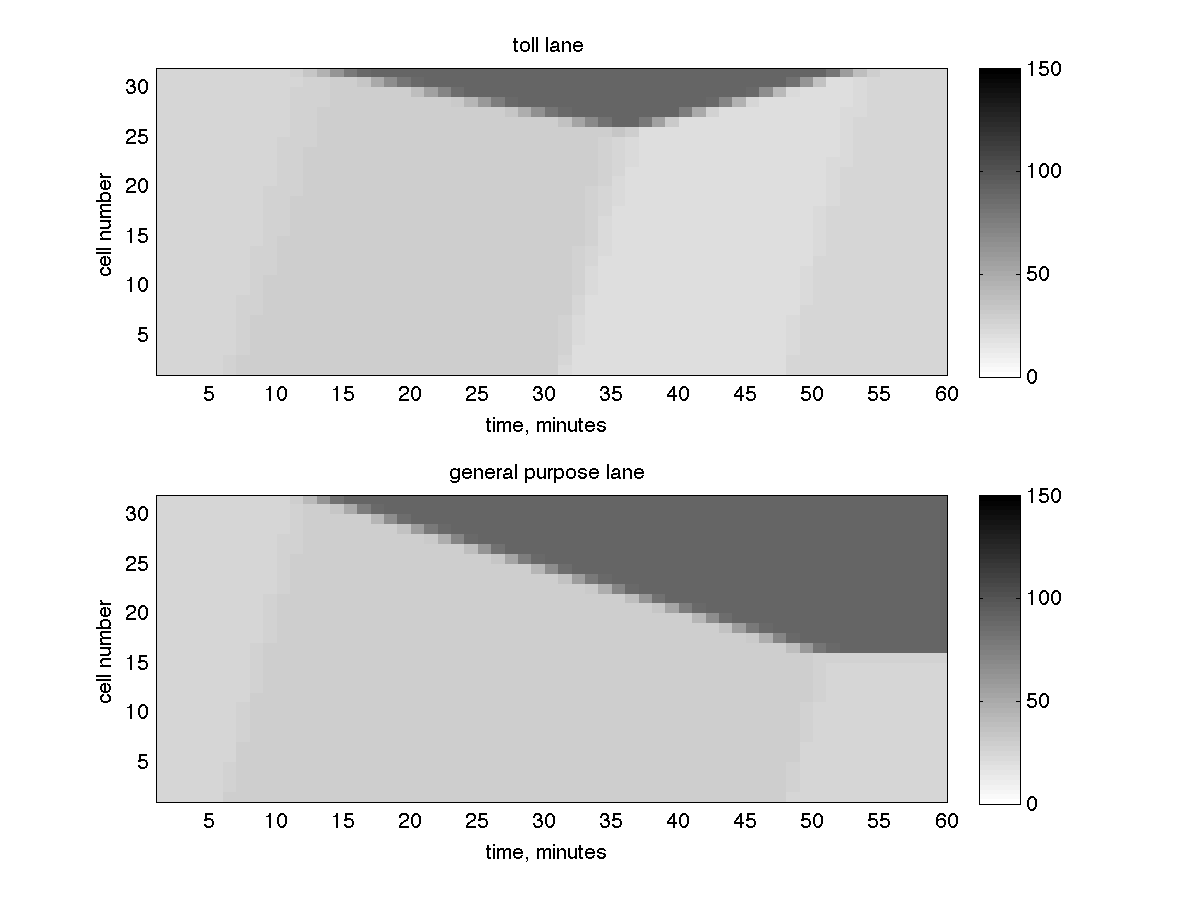}
\caption{Temporarily infeasible incoming flow (scenario 2). Densities.}
\label{fig:scenario-2}
\end{figure}

Now consider a freeway with two entrances (at the beginning and in the middle),
and two exits (in the middle and at the end), see Figure~\ref{fig:freeway-2-2}.
The last link, $K + 1$, is a bottleneck.
The priorities~$p_{i - 1}^{\xi, f}$, $\xi = 1, 2$ and~$p_i^r$
are equal to capacities:
$p_{i - 1}^{\xi, f} = \beta_{i - 1}^f F_{i - 1}^\xi$,
$p_i^r = R_i$.

\begin{figure}[htb]
\centering
\includegraphics{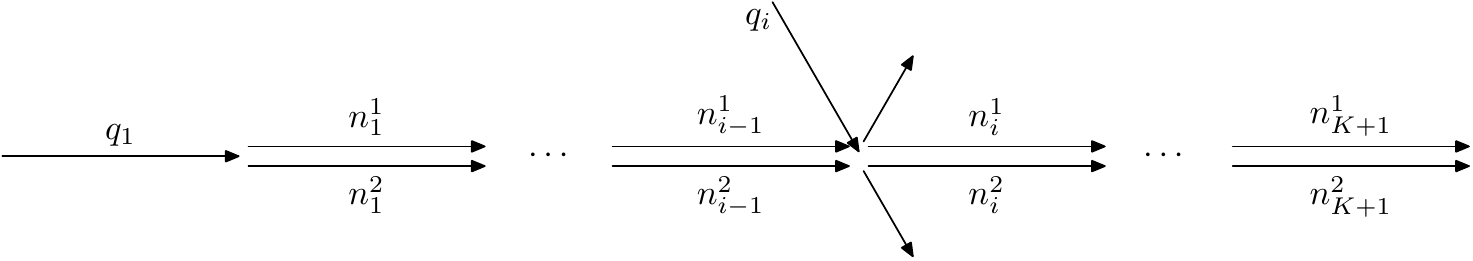}
\caption{Freeway with two entrances and two exits.}
\label{fig:freeway-2-2}
\end{figure}

\paragraph{Scenario 3}
Initially, the toll lane and the general purpose lane
are in the same partially congested equilibrium state.
The incoming flow is feasible, but not strictly feasible.
Flows from both on-ramps (at the beginning and in the middle)
are redistributed,
and the toll lane decongests shifting the extra vehicles to the
general purpose lane.

\begin{figure}[htbp]
\centering
\includegraphics[scale=\scale]{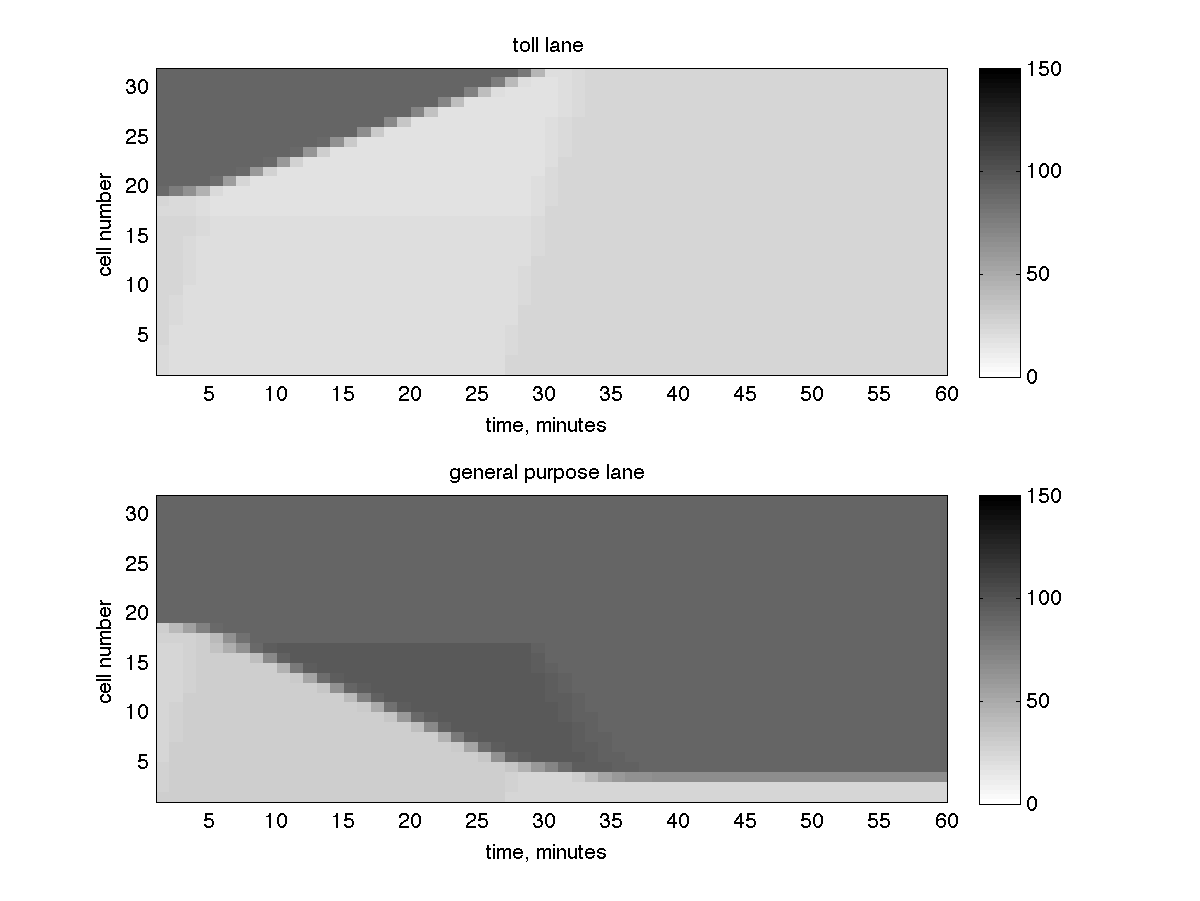}
\caption{Reducing congestion in the toll lane of a freeway with two entrances (scenario~3). Densities.}
\label{fig:scenario-3}
\end{figure}

\section{Tolls as Actuators}\label{sec_tolls}
Now that at a given entrance $i$, the split ratio controller
computed the desired values $\alpha^1_i$ and $\alpha^2_i$,
our goal is to enforce them.
It is done by setting the price for vehicles entering the toll lane.
We consider two ways of setting the price: (1) using the known
Value of Time distribution; and (2) setting up an auction.
These are described next.

\subsection{Value of Time Distribution}\label{subsec-vot}
The Value of Time (VoT) is the marginal rate of substitution of travel time for
money in a travelers' indirect utility function.
In essence, this makes it the amount that a traveler would be willing to pay 
in order to save time, or the amount he would accept as a compensation for the lost
time.
The VoT varies considerably from traveler to traveler and depends upon
the purpose of the journey.
The estimation of the VoT was studied in 
\cite{smallyan01, steimetz05, brownstone05, train09}.

We assume that the VoT distribution is known, $\nu(\pi)$, where
$\pi$ represents the price per time unit.
For a given entrance $i$,
we find the desired price of the time unit, $\pi^\star$ from the equation
\begin{equation}
\int_0^{\pi^\star}\nu(\pi)d\pi = \alpha^2_i.
\label{eq-vot}
\end{equation}
Figure \ref{fig-vot} illustrates the VoT pricing mechanism.
\begin{figure}[htb]
\centering
\includegraphics[scale=.6]{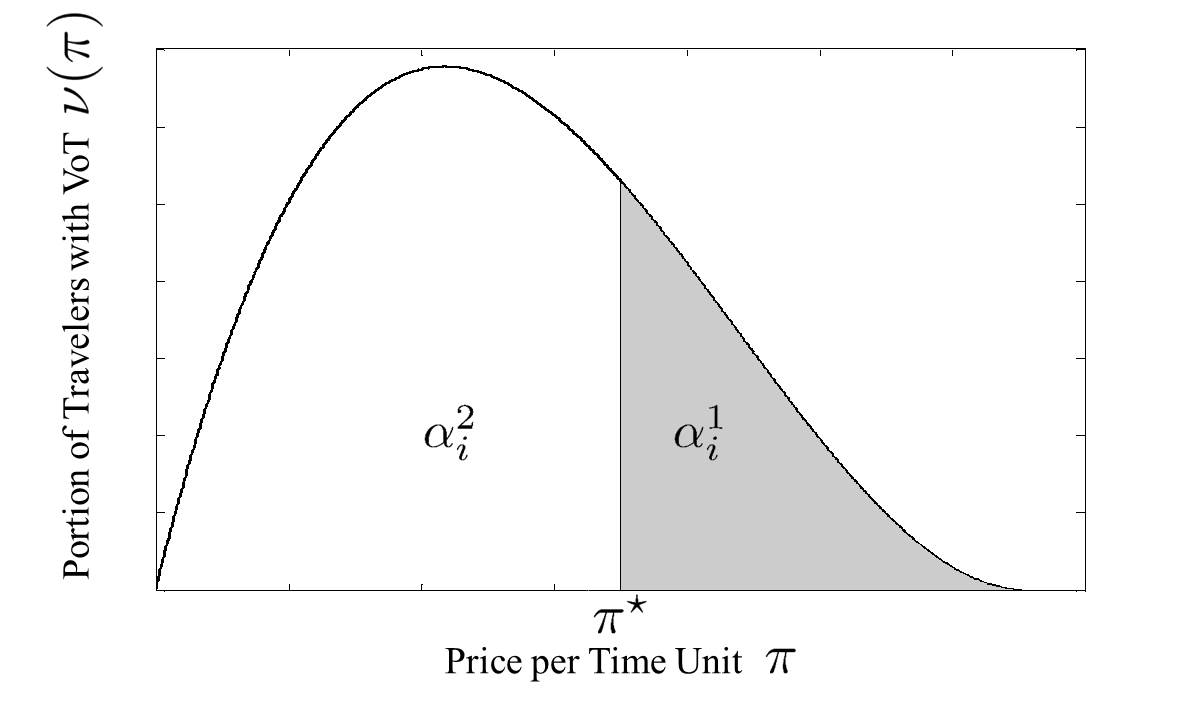}
\caption{Value of Time distribution --- finding price per time unit $\pi$.}
\label{fig-vot}
\end{figure}

To continuously estimate the VoT distribution, we need to know the difference
in travel time between the general purpose and the toll lanes, $\tau$;
the number of vehicles in the general purpose lane, $n^2$;
the price set for the toll lane, $\pi^\star\tau$; and
the amount of money collected, $T$, every certain time period.\footnote{Time
period should be in the order of minutes.}
This will give us enough information to infer the portion of travelers,
whose VoT is higher or equal than $\pi^\star$:
$n^1 = \frac{T}{\pi^\star\tau}$, and
\begin{equation}
\int_{\pi^\star}^\infty\nu(\pi)d\pi=\frac{n^1}{n^1+n^2}.
\label{eq-vot-calib}
\end{equation}

Assuming price $\pi^\star$ varies, we could eventually, estimate
$\nu(\pi^\star)$.
It is rather crude approach, since in different time of day
travelers' VoT is different, but it conveys the idea.

In the case of HOT lanes, where besides paying customers there may be
high occupancy vehicles, we would need to know additionally the number of
vehicles in the HOT lane, $n^1$.
Then, equation (\ref{eq-vot-calib}) is to be slightly modified:
\begin{equation}
\int_{\pi^\star}^\infty\nu(\pi)d\pi=\frac{T}{\pi^\star\tau(n^1+n^2)}.
\label{eq-vot-calib2}
\end{equation}
The main shortcoming of the VoT-based price setting is the difficulty
of estimating the VoT distribution.
If the estimate is too rough, we may end up with underutilized, if
the price is too high, or overly congested, if the price is too low,
toll lane, which will hinder the performance of the overall system
reducing the total output flow.

To reduce the inaccuracy of the VoT based toll calculation, the continuing
calibration of the willingness to pay must be performed.
Such calibration would use a discrete choice model.
Logit model based calibration is described in \cite{michalaka09, lou11}.

\subsection{Auction}\label{subsec-auction}
The alternative method of setting the price for a toll lane is an auction.
The auction based scheme applied to a cordon area congestion pricing
was described in \cite{teodorovic08}.
Here we describe the auction mechanism for managed toll lane
based on the idea of \cite{goldberg06}.
At time $t$, at the entrance $i$, $H=r_i^d(t)$ travelers
make bids $b_1, \dots, b_H$.
Without loss of generality, we assume $b_1\geq\dots\geq b_H$.
Let $h^\star = \mbox{round}(\alpha_i^1H)$.
Vehicles with bids $b_1, \dots, b_{h^\star}$ will be let into 
the toll lane, and each of them will pay $b_{h^\star}$ generating
$b_{h^\star}h^\star$ total revenue for the entrance $i$ at time $t$.

The main advantage of the auction approach is its deterministic outcome:
for each decision point $i$, we can achieve the desired
split ratio coefficients $\alpha_i^1$, $\alpha_i^2$ exactly.

The suggested auction mechanism allows variations.
For instance, parameter $h^\star$ may be chosen maximizing the revenue
without compromising the quality of service:
\[
h^\star = \arg \max_{h\leq\alpha_i^1 H}hb_h .
\]
At the time of this writing, the deployment of an auction as a mechanism
for setting tolls presents certain technical difficulties.
In the coming years, however, the active development of cooperative taxis,
connected and autonomous vehicles and wireless communications will enable
passengers
to determine the amount they would be ready to spend reducing their travel
time at the start or during the trip.
Thus, the auction mechanism for buying a place in the toll lane guaranteeing 
the required time saving may become realistic and practical.

\section{Application: Interstate 680 South in California}\label{sec_application}
In 2013-14 Caltrans\footnote{California Department of Transportation}
conducted a Corridor System Management Plan (CSMP) study for
the Interstate 680 corridor in Contra Costa County \cite{680csmp}.
Certain improvements were considered, including ramp metering, adding
auxiliary lanes and, notably, converting existing HOV lanes
to HOT.\footnote{HOV lane is a lane for High Occupancy Vehicles. 
Typical minimum vehicle occupancy level
for HOV lanes in the U.S. is 2 (2+HOV) or sometimes 3 (3+HOV).
Presently, I-680 corridor has 2+HOV lanes.
HOT stands for High Occupancy or Tolled. 
HOT lane is free for HOVs, others must pay a toll.}
Modeling of improvement scenarios was done at UC Berkeley PATH using
simulation Tools for Operational Planning (TOPL) \cite{topl}.

In this Section we present the simulation for a 10-mile segment of I-680
South corridor extending from Martinez through Concord and Pleasant Hill to
Walnut Creek, shown in Figure \ref{fig-680s}.
It goes from postmile 56 to postmile 46.
An HOV lane spans most of this segment.
The start and end of the HOV lane are shown on the map.

\begin{figure}[htb]
\centering
\includegraphics[scale=.7]{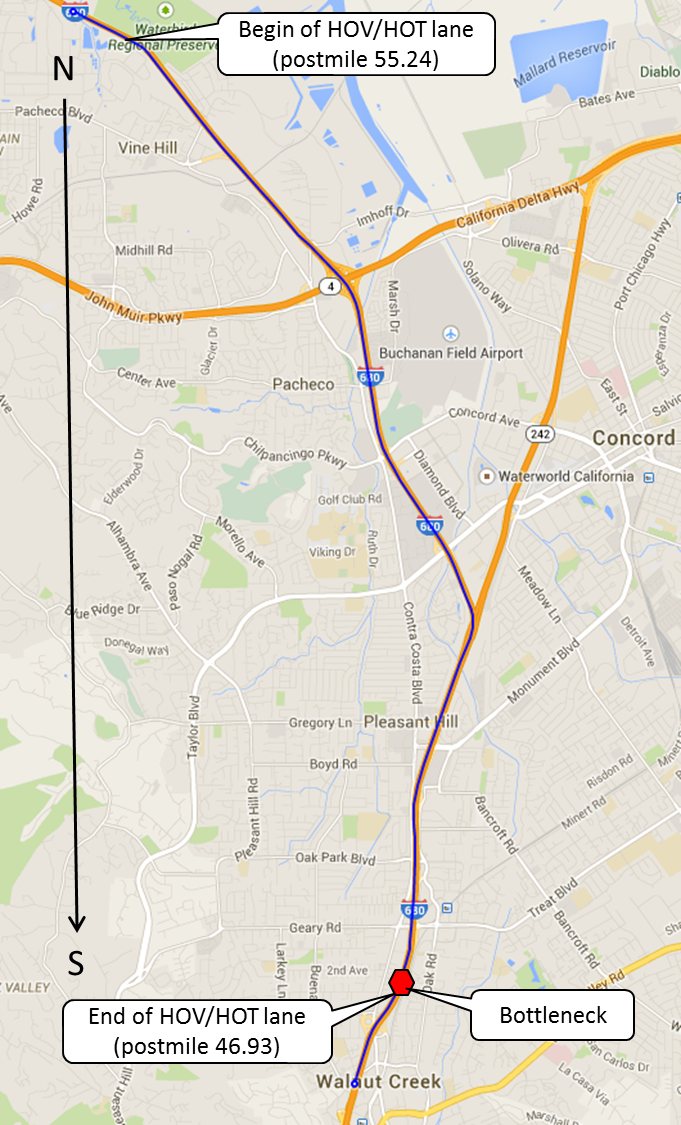}
\caption{Map of the simulated I-680S corridor.}
\label{fig-680s}
\end{figure}

The data collection effort of the CSMP project revealed a severe bottleneck
during the AM peak hours located near the end of the HOV lane.
This bottleneck generates large congestion in the GP lane, while
the HOV lane stays in free flow and is underutilized.

\begin{figure}[htb]
\centering
\includegraphics[scale=.7]{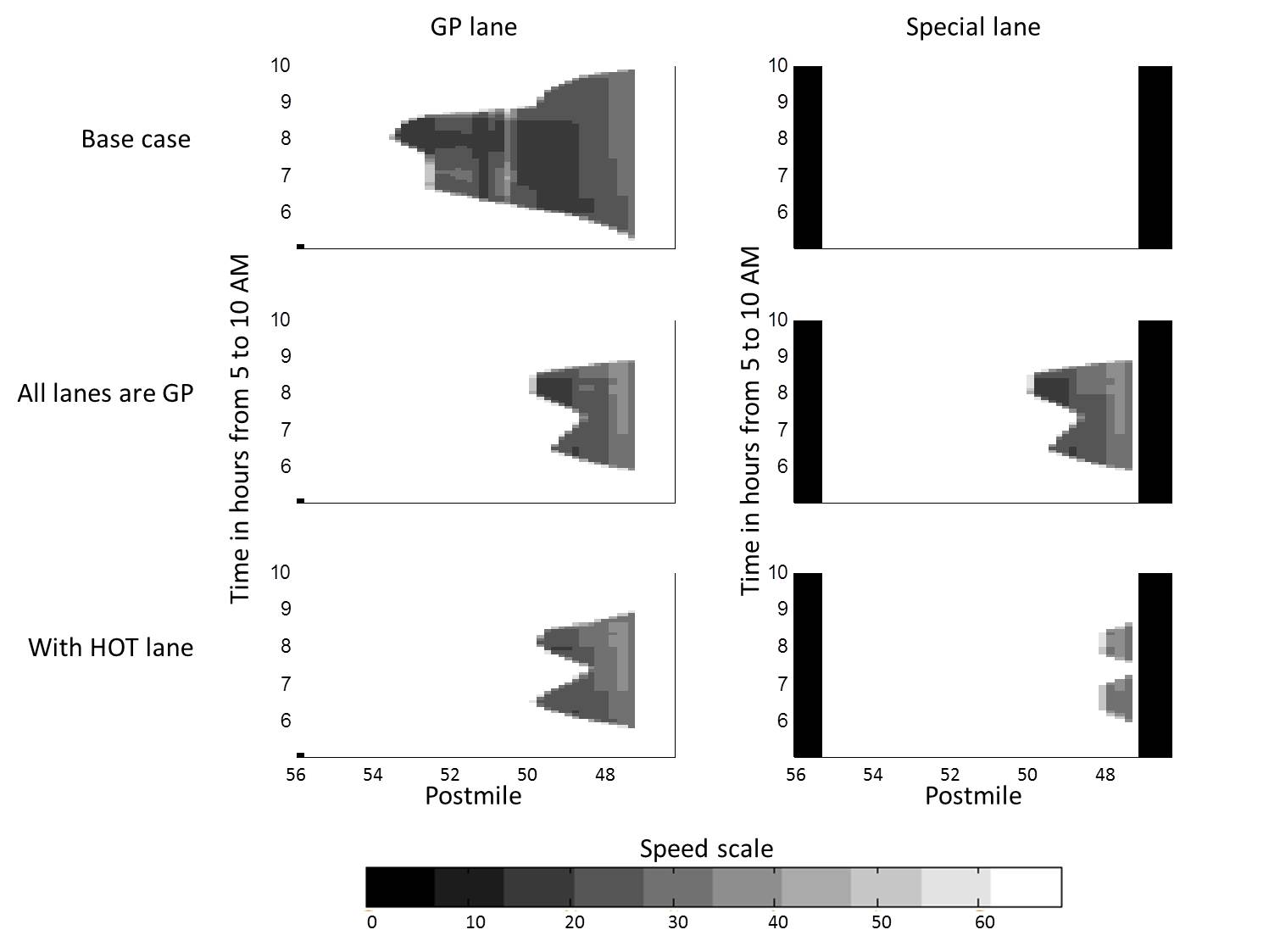}
\caption{Simulation speed contours of the AM peak.
Three cases are considered: base case, or what is currently observed with
an existing HOV lane;
a case when an HOV lane is treated as GP; and
a case when an HOV lane is converted to HOT.
Speed is given in miles per hour.}
\label{fig-sim-contours}
\end{figure}

A macroscopic CTM-based simulation model for the I-680 corridor
was built and calibrated using TOPL and measurement
data collected in the CSMP study.
Three scenarios were simulated:
\begin{enumerate}
\item \emph{base case} that reproduces existing conditions;
\item \emph{HOV lane as GP} that shows how the corridor would
perform if the HOV lane were treated as an extra GP lane;
and
\item \emph{HOV lane as HOT}, where the HOT lane was modeled
instead of the HOV lane using the controller described in this paper.
\end{enumerate}
Whereas the first scenario represents the reality, scenarios
2 and 3 are hypothetical and were explored as part of the planning
exercise.
The input demand was assumed the same in all three scenarios.
Morning peak hours were simulated: from 5 to 10 AM.
This is when travelers currently experience large delays at that
segment of I-680S.

Speed contour maps resulting from the three simulations are
shown in Figure \ref{fig-sim-contours}.
Contours on the left show the speed dynamics in the GP lane,
and contours on the right correspond to the special lane,
which in scenario 1 is HOV, in scenario 2 is GP and in scenario 3 is HOT.
Black stripes on the left and on the right of the special lane contours
correspond to locations, where there is no special lane.
Tables \ref{tab-vmt} and \ref{tab-delay} contain Vehicle Miles Traveled (VMT)
and delay values from the simulated scenarios.

\begin{table}[htb]
\centering
\begin{tabular}{|l|ccc|}\hline
Case & GP lane VMT & Special lane VMT & Total VMT \\ 
\hline
Base case with existing HOV lane & 22502 & 2742 & 25244\\
Special lane is treated as GP & 21087 & 4157 & 25244\\
HOV lane is converted to HOT & 21925 & 3319 & 25244\\
\hline
\end{tabular}
\caption{Vehicle Miles Traveled in the period from 5 to 10 AM
in the three simulated cases.}
\label{tab-vmt}
\end{table}

\begin{table}[htb]
\centering
\begin{tabular}{|l|ccc|}\hline
Case & GP lane delay & Special lane delay & Total delay \\ 
\hline
Base case with existing HOV lane & 250 & 0 & 277\\
Special lane is treated as GP & 77 & 14 & 91\\
HOV lane is converted to HOT & 75 & 2 & 77\\
\hline
\end{tabular}
\caption{Vehicle-hours of delay in the period from 5 to 10 AM
in the three simulated cases.
Total delay in the last column includes delay from on-ramp queues, which
exists only in the base case.}
\label{tab-delay}
\end{table}

In the base case scenario we observe large congestion in the GP lane,
while the HOV lane has no congestion at all.
Opening the HOV lane to everyone (scenario 2) helps a lot, as is evident
from the delay table (Table \ref{tab-delay}), but congestion
is not eliminated completely, and now both GP and HOV lanes
have the same congestion pattern, so travelers have no mode choice.
Not surprisingly, the HOT lane scenario shows the best performance
of the three both in congestion mitigation and keeping the
special lane as free as possible.

One interesting detail to notice is the slight reduction of the GP lane delay
in scenario 3 compared to scenario 2.
That happens despite the fact that the GP lane has more vehicles in scenario 3
than in scenario 2 (Table \ref{tab-vmt}).
This happens because the HOT controller try to keep the
flow in the GP lane as close to capacity as possible, and thus more
vehicles travel with higher speed.

\section{Conclusion}\label{sec_conclusion}
In this paper we described the toll lane control algorithm for freeways
in the context of the link-node CTM with the modified node model.
Its design was presented in two stages: (1) the supply control computing
the desired split ratios for the incoming flows with the goal of
keeping the toll lane in free flow as long as and as much as possible,
yet fully utilized; and (2) the demand control, setting the price
so that the flows between the toll and the general purpose lanes
were distributed according to the computed split ratios.

We described the equilibrium of the traffic model and pointed out
that activating the toll lane is meaningful only in the partially
congested equilibrium, when the general purpose lane has enough
storage space to accommodate extra vehicles from the toll lane
without reducing the total output flow of the system.
In free flow state with feasible demand there is obviously no need
for the toll lane, as it does not provide the benefit of the
shorter travel time.
When the freeway is in the fully congested equilibrium, the deployment
of a toll lane is not recommended either, as it can only harm
the overall system performance by oversaturating the general purpose
lane, creating queues at entrances
and reducing the total output flow of the system.

The split ratio control algorithm presented in the paper behaves as follows:
\begin{itemize}
\item in free flow state with feasible demand it is essentially non-active 
letting the mainline operate as one piece;
\item in the case of infeasible demand, the excess flow is directed into
the general purpose lane, letting it to congest first;
\item once the general purpose lane is fully congested, the algorithm lets
the toll lane to congest;
\item if it comes to the point when the toll lane is fully congested,
the algorithm is deactivated, as there is nothing it can improve
at this point, so the mainline returns back to one piece;
\item once the demand drops again, and the freeway starts to decongest,
the algorithm first brings the toll lane into the free flow state,
and then the general purpose lane.
\end{itemize}
All these manipulations with flows are done through split ratios at
the end nodes of the entrance links.
The control algorithm computes these split ratios.

The split ratio control algorithm can be extended to the case of HOT lanes.
For that, however, one has to introduce and deal with multiple
vehicle types, as described in \cite{kurmur09}.
In this case, two vehicle types will be needed: High Occupancy Vehicles (HOVs)
and Single Occupancy Vehicles (SOVs).
The HOV flow won't be controlled, but will be governed by the algorithm of
finding the path of the least resistance,\footnote{At nodes where flow
exchanges are allowed between the general purpose and the HOT lane, the
HOVs will choose their next downstream link based on higher speed or lower 
density.}
presented in \cite{kurmur09}, whereas the SOVs will be subject to the toll 
control,
and the available space in the general purpose and the HOT lanes will be 
computed accounting for both vehicle types.

Finally, we discussed two mechanisms for price setting --- the VoT
distribution and the auction.
The former is commonly used, but provides only rough estimates of the actual
traveler behavior, resulting mostly in the underutilization or
oversaturation of the toll lane.
The latter mechanism is hardly deployable at current time due to technical
difficulties, but in the foreseeable future with the emerging technologies,
especially, those of autonomous vehicles and ride shares,
it looks promising.

\section*{Acknowledgement}
This work is supported by the California Department of Transportation
(Caltrans) under the Connected Corridors program and by
the Russian Foundation for Basic Research (grants 13-01-90419
and 12-01-00261-a).

\bibliographystyle{plain}
\bibliography{traffic}

\end{document}